\documentclass{amsart}[12 pt]

\usepackage{latexsym}

\usepackage{amsmath,amssymb, amsfonts}
\usepackage{verbatim}
\usepackage{graphicx}
\usepackage{graphics}
\usepackage{geometry}
\usepackage{booktabs}
\usepackage{color}

\newtheorem{theorem}{Theorem}[section]
\newtheorem{definition}{Definition}[section]
\newtheorem{corollary}{Corollary}[section]
\newtheorem{lemma}{Lemma}[section]
\newtheorem{remark}{Remark}[section]

\newcommand{\diag}{\operatorname{diag}}

\numberwithin{equation}{section}

\allowdisplaybreaks

\begin{document}
\markboth{}{}

\title[Dynamic hedging under stochastic volatility] {A general Multidimensional Monte Carlo Approach for Dynamic Hedging under stochastic volatility}

\author{Dorival Le\~ao}

\address{Departamento de Matem\'atica Aplicada e Estat\'istica. Universidade de S\~ao
Paulo, 13560-970, S\~ao Carlos - SP, Brazil} \email{leao@icmc.usp.br}\email{vinicns@icmc.usp.br}

\author{Alberto Ohashi}

\address{Departamento de Matem\'atica, Universidade Federal da Para\'iba, 13560-970, Jo\~ao Pessoa - Para\'iba, Brazil} \email{alberto.ohashi@pq.cnpq.br}

\author{Vin\'{\i}cius Siqueira}

\thanks{We would like to thank Bruno Dupire and Francesco Russo for stimulating discussions and several suggestions about the numerical algorithm proposed in this work. We also gratefully acknowledge the computational support from LNCC (Laborat\'orio Nacional de Computa\c{c}\~ao Cient\'ifica - Brazil). The second author was supported by CNPq grant 308742.}
\date{\today}


\begin{center}
\end{center}

\begin{abstract}
In this work, we introduce a Monte Carlo method for the dynamic hedging of general European-type contingent claims in a multidimensional Brownian arbitrage-free market. Based on bounded variation martingale approximations for Galtchouk-Kunita-Watanabe decompositions, we propose a feasible and constructive methodology which allows us to compute \textit{pure hedging strategies} w.r.t arbitrary square-integrable claims in incomplete markets. In particular, the methodology can be applied to quadratic hedging-type strategies for fully path-dependent options with stochastic volatility and discontinuous payoffs. We illustrate the method with numerical examples based on generalized F\"{o}llmer-Schweizer decompositions, locally-risk minimizing and mean-variance hedging strategies for vanilla and path-dependent options written on local volatility and stochastic volatility models.
\end{abstract}

\maketitle

\section{Introduction}
\subsection{Background and Motivation}\label{back}
Let $(S,\mathbf{F}, \mathbb{P})$ be a financial market composed by a continuous $\mathbf{F}$-semimartingale $S$ which represents a discounted risky asset price process, $\mathbf{F}=\{\mathbf{F}_t;0\le t\le T \}$ is a filtration which encodes the information flow in the market on a finite horizon $[0,T]$, $\mathbb{P}$ is a physical probability measure and $\mathcal{M}^e$ is the set of equivalent local martingale measures. Let $H$ be an $\mathbf{F}_T$-measurable contingent claim describing the net payoff whose the trader is faced at time $T$. In order to hedge this claim, the trader has to choose a dynamic portfolio strategy.

Under the assumption of an arbitrage-free market, the classical Galtchouk-Kunita-Watanabe~(henceforth abbreviated by GKW) decomposition yields

\begin{equation}\label{frep}
H = \mathbb{E}_{\mathbb{Q}}[H] + \int_0^T \theta^{H,\mathbb{Q}}_\ell dS_\ell + L^{H,\mathbb{Q}}_T\quad\text{under}~\mathbb{Q}\in \mathcal{M}^e,
\end{equation}
where $L^{H,\mathbb{Q}}$ is a $\mathbb{Q}$-local martingale which is strongly orthogonal to $S$ and $\theta^{H,\mathbb{Q}}$ is an adapted process.

The GKW decomposition plays a crucial role in determining optimal hedging strategies in a general Brownian-based market model subject to stochastic volatility. For instance, if $S$ is a one-dimensional It\^o risky asset price process which is adapted to the information generated by a two-dimensional Brownian motion $W = (W^1,W^2)$, then there exists a two-dimensional adapted process $\phi^{H,\mathbb{Q}}:=(\phi^{H,1},\phi^{H,2})$ such that

$$H = \mathbb{E}_{\mathbb{Q}}[H] + \int_0^T\phi^{H,\mathbb{Q}}_tdW_t,$$
which also realizes
\begin{equation}\label{basicob}
\theta^{H,\mathbb{Q}}_t = \phi^{H,1}_t [S_t \sigma_t]^{-1},\quad L^{H,\mathbb{Q}}_t = \int_0^t \phi^{H,2}dW^2_s;~0\le t\le T.
\end{equation}

In the complete market case, there exists a unique $\mathbb{Q}\in \mathcal{M}^e$ and in this case, $L^{H,\mathbb{Q}}=0$, $\mathbb{E}_{\mathbb{Q}}[H]$ is the unique fair price and the hedging replicating strategy is fully described by the process $\theta^{H,\mathbb{Q}}$. In a general stochastic volatility framework, there are infinitely many GKW orthogonal decompositions parameterized by the set $\mathcal{M}^e$ and hence one can ask if it is possible to determine the notion of non-self-financing optimal hedging strategies solely based on the quantities \eqref{basicob}. This type of question was firstly answered by F\"{o}llmer and Sonderman~\cite{Sondermann1986} and later on extended by Schweizer~\cite{Schweizer1991} and F\"{o}llmer and Schweizer~\cite{Follmer1991} through the existence of the so-called F\"{o}llmer-Schweizer decomposition which turns out to be equivalent to the existence of locally-risk minimizing hedging strategies. The GKW decomposition under the so-called minimal martingale measure constitutes the starting point to get locally risk minimizing strategies provided one is able to check some square-integrability properties of the components in~(\ref{frep}) under the physical measure. See e.g~\cite{Heath} and~\cite{Schweizer2001} for details and other references therein. Orthogonal decompositions without square-integrability properties can also be defined in terms of the the so-called generalized F\"{o}llmer-Schweizer decomposition~(see e.g Schweizer~\cite{Schweizer1995}).

In contrast to the local-risk minimization approach, one can insist in working with self-financing hedging strategies which give rise to the so-called mean-variance hedging methodology. In this approach, the spirit is to minimize the expectation of the squared hedging error over all initial endowments $x$ and all suitable admissible strategies $\varphi \in \Theta$:

\begin{equation}\label{2opt}
\inf_{\varphi\in \Theta, x\in \mathbb{R}}\mathbb{E}_{\mathbb{P}}\Big|H- x - \int_0^T \varphi_tdS_t\Big|^2.
\end{equation}
The nature of the optimization problem \eqref{2opt} suggests to work with the subset $\mathcal{M}^e_2:=\{\mathbb{Q}\in\mathcal{M}^e; \frac{d\mathbb{Q}}{d\mathbb{P}}\in L^2(\mathbb{P})\}$. Rheinlander and Schweizer~\cite{rheinlander}, Gourieroux, Laurent and Pham~\cite{pham_1} and Schweizer~\cite{Schweizer_1996} show that if $\mathcal{M}^e_2\neq \emptyset$ and $H\in L^2(\mathbb{P})$ then the optimal quadratic hedging strategy exists and it is given by $\big(\mathbb{E}_{\tilde{\mathbb{P}}}[H],\eta^{\tilde{\mathbb{P}}}\big)$, where

\begin{equation}\label{opmvhs}
\eta^{\tilde{\mathbb{P}}}_t :=\theta^{H,\tilde{\mathbb{P}}}_t - \frac{\tilde{\zeta}_t}{\tilde{Z}_t}\Bigg(V^{H,\tilde{\mathbb{P}}}_{t-} - \mathbb{E}_{\tilde{\mathbb{P}}}[H] - \int_0^t\eta^{\tilde{\mathbb{P}}}_\ell dS_\ell \Bigg);~0\le t\le T.
\end{equation}
Here $\theta^{H,\tilde{\mathbb{P}}}$ is computed in terms of $\tilde{\mathbb{P}}$, the so-called variance optimal martingale measure, $\tilde{\zeta}$ realizes

\begin{equation}\label{radonvomm}
\tilde{Z}_t:=\mathbb{E}_{\tilde{\mathbb{P}}}\Bigg[\frac{d\tilde{\mathbb{P}}}{d\mathbb{P}}\Big|\mathbf{F}_t \Bigg] = \tilde{Z}_0 + \int_0^t\tilde{\zeta}_\ell dS_\ell;~0\le t\le T,
\end{equation}
and $V^{H,\tilde{\mathbb{P}}}:=\mathbb{E}_{\tilde{\mathbb{P}}}[H|\mathbf{F}_\cdot]$ is the value option price process under $\tilde{\mathbb{P}}$. See also Cern\'y and Kallsen~\cite{kallsen_2} for the general semimartingale case and the works~\cite{kallsen_1},~\cite{kramkov1} and~\cite{kramkov2} for other utility-based hedging strategies based on GKW decompositions.

Concrete representations for the \textit{pure hedging strategies} $\{\theta^{H,\mathbb{Q}}; \mathbb{Q}=\hat{\mathbb{P}},\tilde{\mathbb{P}}\}$ can in principle be obtained by computing cross-quadratic variations $d[V^{H,\mathbb{Q}}, S]_t/d[S,S]_t$ for $\mathbb{Q}\in \{\tilde{\mathbb{P}},\hat{\mathbb{P}}\}$. For instance, in the classical vanilla case, pure hedging strategies can be computed by means of the Feynman-Kac theorem~(see e.g~Heath, Platen and Schweizer~\cite{Heath}). In the path-dependent case, the obtention of concrete computationally efficient representations for $\theta^{H,\mathbb{Q}}$ is a rather difficult problem. Feynman-Kac-type arguments for fully path-dependent options mixed with stochastic volatility typically face not-well posed problems on the whole trading period as well as highly degenerate PDEs arise in this context. Generically speaking, one has to work with non-Markovian versions of the Feynman-Kac theorem in order to get robust dynamic hedging strategies for fully path dependent options written on stochastic volatility risky asset price processes.

In the mean variance case, the only quantity in \eqref{opmvhs} not related to GKW decomposition is $\tilde{Z}$ which can in principle be expressed in terms of the so-called fundamental representation equations given by Hobson~\cite{hobson} and Biagini, Guasoni and Pratelli~\cite{biagini} in the stochastic volatility case. For instance, Hobson derives closed form expressions for $\tilde{\zeta}$ and also for any type of $q$-optimal measure in the Heston model~\cite{heston}. Recently, semi-explicit formulas for vanilla options based on general characterizations of the variance-optimal hedge in Cern\'y and Kallsen~\cite{kallsen_2} have been also proposed in the literature which allow for a feasible numerical implementation in affine models. See~Kallsen and Vierthauer~\cite{kallsen_3} and Cern\'y and Kallsen~\cite{kallsen_4} for some results in this direction. A different approach based on backward stochastic differential equations can also be used in order to get useful characterizations for the optimal mean variance hedging strategies. See e.g~Jeanblanc, Mania, Santacrose and Schweizer~\cite{jeanblanc} and other references therein.

\subsection{Contribution of the current paper.}
In spite of deep characterizations of optimal quadratic hedging strategies and concrete numerical schemes available for vanilla-type options, to our best knowledge no feasible approach has been proposed to tackle the problem of obtaining dynamic optimal quadratic hedging strategies for fully path dependent options written on a generic multidimensional It\^o risky asset price process. In this work, we attempt to solve this problem with a probabilistic approach. The main difficulty in dealing with fully path dependent and/or discontinuous payoffs is the non-Markovian nature of the option value and a priori lack of path smoothness of the pure hedging strategies. Usual numerical schemes based on PDE and martingale techniques do not trivially apply to this context.

The main contribution of this paper is the obtention of flexible and computationally efficient multidimensional non-Markovian representations for generic option price processes which allow for a concrete computation of the associated GKW decomposition $\big(\theta^{H,\mathbb{Q}},L^{H,\mathbb{Q}}\big)$  for $\mathbb{Q}$-square integrable payoffs $H$ with $\mathbb{Q}\in\mathcal{M}^e$. We provide a Monte Carlo methodology capable to compute optimal quadratic hedging strategies w.r.t general square-integrable claims in a multidimensional Brownian-based market model.

This article provides a feasible and constructive method to compute generalized F\"{o}llmer-Schweizer decompositions under full generality. As far as the mean variance hedging is concerned, we are able to compute pure optimal hedging strategies $\theta^{H,\tilde{\mathbb{P}}}$ for arbitrary square-integrable payoffs. Hence, our methodology also applies to this case provided one is able to compute the fundamental representation equations in Hobson~\cite{hobson} and Biagini, Guasoni and Pratelli~\cite{biagini} which is the case for the classical Heston model. In mathematical terms, we are able to compute $\mathbb{Q}$-GKW decompositions under full generality so that the results of this article can also be used to other non-quadratic hedging methodologies where orthogonal martingale representations play an important role in determining optimal hedging strategies.

The starting point of this article is based on weak approximations developed by Le\~ao and Ohashi~\cite{LEAO_OHASHI09} for one-dimensional Brownian functionals. They introduced a one-dimensional space-filtration discretization scheme constructed from suitable waiting times which measure the instants when the Brownian motion hits some a priori levels. In this work, we extend~\cite{LEAO_OHASHI09} to the multidimensional case as follows: More general and stronger convergence results are obtained in order to recover incomplete markets with stochastic volatility. Hitting times induced by multidimensional noises which drive the stochastic volatility are carefully analyzed in order to obtain $\mathbb{Q}$-GKW decompositions under rather weak integrability conditions for any $\mathbb{Q}\in \mathcal{M}^e$. Moreover, a complete analysis is performed w.r.t weak approximations for gain processes by means of suitable non-antecipative discrete-time hedging strategies for square-integrable payoffs, including path-dependent ones.

It is important to stress that the results of this article can be applied to both complete and incomplete markets written on a generic multidimensional It\^o risky asset price process. One important restriction of our methodology is the assumption that the risky asset price process has continuous paths. This is a limitation that we hope to overcome in a future work.

Numerical results based on the standard Black-Scholes, local-volatility and Heston models are performed in order to illustrate the theoretical results and the methodology of this article. In particular, we briefly compare our results with other prominent methodologies based on Malliavin weights (complete market case) and PDE techniques (incomplete market case) employed by Bernis, Gobet and Kohatsu-Higa~\cite{kohatsu} and Heath, Platen and Schweizer~\cite{Heath}, respectively. The numerical experiments suggest that pure hedging strategies based on generalized F\"{o}llmer-Schweizer decompositions mitigate very well the cost of hedging of a path-dependent option even if there is no guarantee of the existence of locally-risk minimizing strategies. We also compare hedging errors arising from optimal mean variance hedging strategies for one-touch options written on a Heston model with nonzero correlation.

The remainder of this paper is structured as follows. In Section \ref{capitulo:modelo}, we fix the notation and we describe the basic underlying market model. In Section~\ref{capitulo:aproximacao}, we provide the basic elements of the Monte Carlo methodology proposed in this article. In Section~\ref{wdhsection}, we formulate dynamic hedging strategies starting from a given GKW decomposition and we translate our results to well-known quadratic hedging strategies. The Monte Carlo algorithm and the numerical study are described in Sections~\ref{capitulo:algoritmo} and \ref{capitulo:resultados}, respectively. The Appendix presents more refined approximations when the martingale representations admit additional hypotheses.

\section{Preliminaries}\label{capitulo:modelo}
Throughout this paper, we assume that we are in the usual Brownian market model with finite time horizon $0\leq T< \infty$ equipped with the stochastic basis $(\Omega, \mathbf{F},\mathbb{P})$ generated by a standard $p$-dimensional Brownian motion $B = \{(B^{(1)}_t,\ldots,B^{(p)}_t); 0\le t\le T \}$ starting from $0$. The filtration $\mathbf{F}:=(\mathbf{F}_t)_{0\leq t\leq T}$ is the $\mathbb{P}$-augmentation of the natural filtration generated by $B$. For a given $m$-dimensional vector $J = (J_1 , \ldots , J_m)$, we denote by $\diag(J)$ the $m \times m$ diagonal matrix whose $\ell$-th diagonal term is $J_{\ell}$. In this paper, for all unexplained terminology concerning general theory of processes, we refer to Dellacherie and Meyer~\cite{dellacherie}.

In view of stochastic volatility models, let us split $B$ into two multidimensional Brownian motions as follows $B^S := (B^{(1)},\ldots,B^{(d)})$ and $B^I := (B^{(d+1)},\ldots,B^{(p)})$. In this section, the market consists of $d + 1$ assets $(d \leq p)$: one riskless asset given by

$$
d S^0_t=r_tS^0_tdt, \quad S^0_0=1; \quad 0\leq t\leq T,
$$
and a $d$-dimensional vector of risky assets $\bar{S}:=(\bar{S}^{1},\ldots,\bar{S}^{d})$ which satisfies the following stochastic differential equation

$$d \bar{S}_t = \diag(\bar{S}_t) \left( b_t dt + \sigma_tdB^S_t \right),\quad \bar{S}_0 = \bar{x}\in \mathbb{R}^d; \quad 0\leq t\leq T.$$

Here, the real-valued interest rate process $r = \{r_t;0\leq t\leq T\}$, the vector of mean rates of return $b := \{b_t=(b^{1}_t,\ldots,b^{d}_t);0\leq t\leq T\}$ and the volatility matrix $\sigma := \{\sigma_t=(\sigma^{ij}_t); 1\leq i\leq d,1\le j\le d,~0\leq t\leq T\}$ are assumed to be predictable and they satisfy the standard assumptions in such way that both $S^0$ and $\bar{S}$ are well-defined positive semimartingales. We also assume that the volatility matrix $\sigma$ is non-singular for almost all $(t, \omega)\in [0,T]\times \Omega$. The discounted price $S := \{S_i:=\bar{S}^i / S^0; i=1,\ldots, d\}$ follows

$$d S_t = \diag(S_t) \left[( b_t - r_t \textbf{1}_d) dt + \sigma_tdB^S_t \right];\quad S_0 = x\in \mathbb{R}^d,~0\leq t\leq T,$$
where $\textbf{1}_d$ is a d-dimensional vector with every component equal to $1$. The market price of risk is given by




$$\psi_t:= \sigma^{-1}_t \left[ b_t - r_t \textbf{1}_d  \right], \quad 0 \leq t \leq T,$$
where we assume
$$\int_{0}^T \| \psi_u\|_{\mathbb{R}^d}^2 du < \infty~a.s.$$

In the sequel, $\mathcal{M}^{e}$ denotes the set of $\mathbb{P}$-equivalent probability measures $\mathbb{Q}$ such that the respective Radon-Nikodym derivative process is a $\mathbb{P}-$martingale and the discounted price $S$ is a $\mathbb{Q}$-local martingale. Throughout this paper, we assume that $\mathcal{M}^e \neq \emptyset$. In our setup, it is well known that $\mathcal{M}^{e}$ is given by the subset of probability measures with Radon-Nikodym derivatives of the form

$$\frac{d\mathbb{Q}}{d\mathbb{P}} := \exp \left[ - \int_{0}^T \psi_u dB^S_u - \int_{0}^T \nu_udB^I_u - \frac{1}{2} \int_{0}^T \big\{\|\psi_u\|^2_{\mathbb{R}^d} + \|\nu_u\|^2_{\mathbb{R}^{p-d}}\big\} du \right],$$
for some $\mathbb{R}^{p-d}$-valued adapted process $\nu$ such that $\int_{0}^T \| \nu_t\|^2_{\mathbb{R}^{p-d}} dt < \infty$ a.s.

\

\noindent \textbf{Example}: The typical example studied in the literature is the following one-dimensional stochastic volatility model

\begin{equation}\label{svmexample}
\left\{\begin{array}{l}
    \displaystyle dS_t = S_t\mu(t,S_t,\sigma_t)dt + S_t\sigma_tdY_t^{(1)}\\
    \displaystyle d\sigma_t^2 = a(t,S_t,\sigma_t)dt + b(t,S_t, \sigma_t)dY^{(2)}_t;~0\le t\le T,
\end{array}\right.
\end{equation}
where $Y^{(1)}$ and $Y^{(2)} $ are correlated Brownian motions with correlation $\rho \in [-1,1]$, $\mu, a$ and $b$ are suitable functions such that $(S, \sigma^2)$ is a well-defined two-dimensional Markov process. All continuous stochastic volatility model commonly used in practice fit into the specification \eqref{svmexample}. In this case, $p=2>d=1$ and we recall that the market is incomplete where the set $\mathcal{M}^e$ is infinity. The dynamic hedging procedure turns out to be quite challenging due to extrinsic randomness generated by the non-tradeable volatility, specially w.r.t to exotic options.


\subsection{GKW Decomposition}
In the sequel, we take $\mathbb{Q} \in \mathcal{M}^{e}$ and we set $W^S:=(W^{(1)},\ldots,W^{(d)})$ and $W^I:=(W^{(d+1)},\ldots, W^{(p)})$ where

\begin{equation}\label{W}
W^{(j)}_t:=  \left\{
\begin{array}{ll}
  \displaystyle B^{(j)}_t + \int_{0}^t \psi^{j}_u du, & j=1, \ldots , d  \\
  \displaystyle B^{(j)}_t + \int_{0}^t \nu^{j}_u du, & j=d+1, \ldots , p;~0\le t\le T,
\end{array}
\right.
\end{equation}
is a standard $p$-dimensional Brownian motion under the measure $\mathbb{Q}$ and filtration $\mathbb{F}:= \{ \mathcal{F}_t;0\le t\le T \}$ generated by $W= (W^{(1)} , \ldots , W^{(p)})$. In what follows, we fix a discounted contingent claim $H$. Recall that the filtration $\mathbb{F}$ is contained in $\mathbf{F}$, but it is not necessarily equal. In the remainder of this article, we assume the following hypothesis.

\

\noindent \textbf{(M)} The contingent claim $H$ is also $\mathcal{F}_T$-measurable.

\
\begin{remark}
Assumption~\textbf{(M)} is essential for the approach taken in this work because the whole algorithm is based on the information generated by the Brownian motion $W$ (defined under the measure $\mathbb{Q}$ and filtration $\mathbb{F}$). As long as the \textit{numeraire} is deterministic, this hypothesis is satisfied for any stochastic volatility model of the form \eqref{svmexample} and a payoff $\Phi (S_t; 0\le t\le T)$ where $\Phi:\mathcal{C}_T\rightarrow \mathbb{R}$ is a Borel map and $\mathcal{C}_T$ is the usual space of continuous paths on $[0,T]$. Hence, \textbf{(M)} holds for a very large class of examples founded in practice.
\end{remark}

For a given $\mathbb{Q}$-square integrable claim $H$, the Brownian martingale representation (computed in terms of $(\mathbb{F},\mathbb{Q})$) yields


$$
H = \mathbb{E}_{{\mathbb{Q}}}[H] + \int_{0}^T \phi^{H, \mathbb{Q}}_u dW_u,
$$
where $\phi^{H, \mathbb{Q}}: = (\phi^{H,\mathbb{Q},1},\ldots,\phi^{H,\mathbb{Q},p} )$ is a $p$-dimensional $\mathbb{F}$-predictable process. In what follows, we set $\phi^{H,\mathbb{Q} , S}: = ( \phi^{H,\mathbb{Q},1},\ldots,\phi^{H,\mathbb{Q},d}  )$, $\phi^{H, \mathbb{Q}, I}: = (\phi^{H,\mathbb{Q},d+1},\ldots,\phi^{H,\mathbb{Q},p} )$ and

\begin{equation}\label{cost}
L^{H, \mathbb{Q}}_t:=\int_{0}^t \phi^{H,\mathbb{Q} , I}_u dW^I_u,\quad \hat{V}_t:=\mathbb{E}_{\mathbb{Q}}[H|\mathcal{F}_t];~0\le t\le T.
\end{equation}
The discounted stock price process has the following $\mathbb{Q}$-dynamics

$$
d S_t =\diag(S_t) \sigma_td W^S_t, \quad S_0 = x, ~  0\leq t\leq T,
$$
and therefore the $\mathbb{Q}$-GKW decomposition for the pair of locally square integrable local martingales $(\hat{V},S)$ is given by

\begin{align}\label{pfsg}
\hat{V}_t & = \mathbb{E}_{\mathbb{Q}}[H] + \int_{0}^t \phi^{H,\mathbb{Q} , S}_u d W^S_u + L^{H, \mathbb{Q}}_t \nonumber\\
          & = \mathbb{E}_{\mathbb{Q}}[H] + \int_{0}^t \theta^{H,\mathbb{Q}}_u d S_u + L^{H, \mathbb{Q}}_t;\quad 0\le t\le T,
\end{align}


\noindent where
\begin{equation}\label{keyprocess}
\theta^{H,\mathbb{Q}}:=\phi^{H,\mathbb{Q},S} \left[\diag(S) \sigma \right]^{-1}.
\end{equation}
The $p$-dimensional process $\phi^{H,\mathbb{Q}}$ which constitutes \eqref{cost} and \eqref{keyprocess} plays a major role in several types of hedging strategies in incomplete markets and it will be our main object of study.

\begin{remark}
If we set $\nu^j= 0$ for $j=d+1,\ldots,p$ and the correspondent density process is a martingale then the resulting minimal martingale measure $\hat{\mathbb{P}}$ yields a GKW decomposition where $L^{H,\hat{\mathbb{P}}}$ is still a $\mathbb{P}$-local martingale orthogonal to the martingale component of $S$ under $\mathbb{P}$. In this case, it is also natural to implement a pure hedging strategy based on $\theta^{H,\hat{\mathbb{P}}}$ regardless the existence of the F\"{o}llmer-Schweizer decomposition. If this is the case, this hedging strategy can be based on the generalized F\"{o}llmer-Schweizer decomposition~(see e.g Th.9 in~\cite{Schweizer1995}).
\end{remark}

\section{The Random Skeleton and Weak Approximations for GKW Decompositions}\label{capitulo:aproximacao}

In this section, we provide the fundamentals of the numerical algorithm of this article for the obtention of hedging strategies in complete and incomplete markets.



\subsection{The Multidimensional Random Skeleton}\label{bo}
At first, we fix once and for all $\mathbb{Q}\in \mathcal{M}^e$ and a $\mathbb{Q}$-square-integrable contingent claim $H$ satisfying~\textbf{(M)}. In the remainder of this section, we are going to fix a $\mathbb{Q}$-Brownian motion $W$ and with a slight abuse of notation all $\mathbb{Q}$-expectations will be denoted by $\mathbb{E}$. The choice of $\mathbb{Q}\in \mathcal{M}^e$ is dictated by the pricing and hedging method used by the trader.

In the sequel, $[\cdot,\cdot]$ denotes the usual quadratic variation between semimartingales and the usual jump of a process is denoted by $\Delta Y_t=Y_t-Y_{t-}$ where $Y_{t-}$ is the left-hand limit of a cadlag process $Y$. For a pair $(a,b)\in \mathbb{R}^2$, we denote $a\vee b:=\max\{a,b\}$ and $a\wedge b:=\min \{a,b\}$. Moreover, for any two stopping times $S$ and $J$, we denote the stochastic intervals $[[S,J [[:=\{(\omega,t); S(\omega) \le t < J(\omega) \}$, $[[S]]:=\{(\omega,t); S(\omega)=t \}$ and so on. Throughout this article, $Leb$ denotes the Lebesgue measure on the interval $[0,T]$.

For a fixed positive integer $k$ and for each $j = 1, 2, \ldots, p$ we define $T^{k,j}_0 := 0$ a.s. and

\begin{equation}\label{stopping_times}
T^{k,j}_n := \inf\{T^{k,j}_{n-1}< t <\infty;  |W^{(j)}_t - W^{(j)}_{T^{k,j}_{n-1}}| = 2^{-k}\}, \quad n \ge 1,
\end{equation}
where $W:=(W^{(1)},\ldots,W^{(p)})$ is the $p$-dimensional $\mathbb{Q}$-Brownian motion as defined in \eqref{W}.

For each $j\in \{1,\ldots,p \}$, the family $(T^{k,j}_n)_{n\ge 0}$ is a sequence of $\mathbb{F}$-stopping times where the increments $\{T^{k,j}_n - T^{k,j}_{n-1}; n\ge 1\}$ is an i.i.d sequence with the same distribution as $T^{k,j}_1$. In the sequel, we define $A^k:=(A^{k,1},\ldots, A^{k,p})$ as the $p$-dimensional step process given componentwise by

\[
A^{k,j}_t := \sum_{n=1}^{\infty}2^{-k}\eta^{k,j}_n1\!\!1_{\{T^{k,j}_n\leq t \}};~0\le t\le T,
\]
where

\begin{equation}\label{sigmakn}
\eta^{k,j}_n:=\left\{
\begin{array}{rl}
1; & \hbox{if} \ W^{(j)}_{T^{k,j}_n} - W^{(j)}_{T^{k,j}_{n-1}} = 2^{-k} ~ ~ \mbox{and} ~ ~ T^{k,j}_n < \infty \\
-1;& \hbox{if} \ W^{(j)}_{T^{k,j}_n} - W^{(j)}_{T^{k,j}_{n-1}} = -2^{-k} ~ ~ \mbox{and}~ ~ T^{k,j}_n < \infty \\
0; & \hbox{if} \ T^{k,j}_n = \infty.
\end{array}
\right.
\end{equation}
for $k,n\ge 1$ and $j=1, \ldots , p$. We split $A^k$ into $(A^{S,k},A^{I,k})$ where $A^{S,k}$ is the $d$-dimensional process constituted by the first $d$ components of $A^k$ and $A^{I,k}$ the remainder $p-d$-dimensional process. Let $\mathbb{F}^{k,j} := \{ \mathcal{F}^{k,j}_t : 0 \leq t\le T \} $ be the natural filtration generated by $\{A^{k,j}_t; 0\leq t \le T\}$. One should notice that $\mathbb{F}^{k,j}$ is a discrete-type filtration in the sense that
\[
\mathcal{F}^{k,j}_t = \bigvee_{\ell=0}^{\infty} \Big( \mathcal{F}^{k,j}_{T^{k,j}_\ell} \cap \{T^{k,j}_{\ell} \le t < T^{k,j}_{\ell+1}\} \Big),~0\le t\le T,
\]
where $\mathcal{F}^{k,j}_0 = \{\Omega, \emptyset \}$ and $\mathcal{F}^{k,j}_{T^{k,j}_m}=\sigma(T^{k,j}_1, \ldots, T^{k,j}_m, \eta^{k,j}_1, \ldots, \eta^{k,j}_m)$ for $m\ge 1$ and $j=1,\ldots, p$. Here, $\bigvee$ denotes the smallest sigma-algebra generated by the union. One can easily check that $\mathcal{F}^{k,j}_{T^{k,j}_m} = \sigma(A^{k,j}_{s\wedge T^{k,j}_m}; s \ge 0)$ and hence

$$\mathcal{F}^{k,j}_{T^{k,j}_m}=\mathcal{F}^{k,j}_t~a.s~\text{on}~\big\{T^{k,j}_m \le t < T^{k,j}_{m+1}\big\}.$$
With a slight abuse of notation we write $\mathcal{F}^{k,j}_t$ to denote its $\mathbb{Q}$-augmentation satisfying the usual conditions.

Let us now introduce the multidimensional filtration generated by $A^k$. Let us consider $\mathbb{F}^k := \{\mathcal{F}^k_t ; 0 \leq t \leq T\}$ where $\mathcal{F}^k_t := \mathcal{F}^{k,1}_t\otimes\mathcal{F}^{k,2}_t\otimes\cdots\otimes\mathcal{F}^{k,p}_t$ for $0\le t\le T$. Let $\mathcal{T}^k:=\{T^k_m; m\ge 0\}$ be the order statistics obtained from the family of random variables $\{T^{k,j}_\ell; \ell\ge 0 ;j=1,\ldots,p\}$. That is, we set $T^k_0:=0$,

\begin{equation}\label{difst}
T^k_1:= \inf_{\substack {1\le j\le p\\ m\ge 1} }\Big\{T^{k,j}_m \Big\},\quad T^k_n:= \inf_{\substack {1\le j\le p\\ m\ge 1} } \Big\{T^{k,j}_m ; T^{k,j}_m \ge T^{k}_1 \vee \ldots \vee T^k_{n-1}\Big\}
\end{equation}
for $n\ge 1$. In this case, $\mathcal{T}^k$ is the partition generated by all stopping times defined in \eqref{stopping_times}. The finite-dimensional distribution of $W^{(j)}$ is absolutely continuous for each $j=1,\ldots,p$ and therefore the elements of $\mathcal{T}^k$ are almost surely distinct for every $k\ge 1$. Moreover, the following result holds true.

\begin{lemma}\label{st}
For every $k\ge 1$, the set $\mathcal{T}^k$ is an exhaustive sequence of $\mathbb{F}^k$-stopping times such that $\sup_{n\ge 1}|T^k_n-T^k_{n-1}|\rightarrow 0$ in probability as $k\rightarrow \infty$.
\end{lemma}
\begin{proof}
The following obvious estimate holds

\[
\sup_{n\ge 1}|T^k_n-T^k_{n-1}| \leq \max_{1\le j\le p}\sup_{n\ge 1}|T^{k,j}_n-T^{k,j}_{n-1}| \rightarrow 0,
\] in probability as $k\rightarrow \infty$ and $T^k_n \rightarrow \infty$ a.s as $n\rightarrow \infty$ for each $k\ge 1$. Let us now prove that $\mathcal{T}^k=\{T^k_n; n\ge 0 \}$ is a sequence of $\mathbb{F}^k$-stopping times. In order to show this, we write $(T^k_n)_{n\ge 0}$ in a different way. This sequence can be defined recursively as follows

\[
T^k_0 = 0, \quad T^k_1 = \inf \{ t > 0 ; \parallel A^k_t \parallel_{\Bbb{R}^p} = 2^{-k} \},
\] where $\parallel \cdot \parallel_{\Bbb{R}^p}$ denotes the $\Bbb{R}^p$-maximum norm. Therefore, $T^k_1$ is an $\mathbb{F}^k$-stopping time. Next, let us define a family of $\mathcal{F}^k_{T^k_1}$-random variables related to the index $j$ which realizes the hitting time $T^k_1$ as follows

\[
\ell^{k,j}_1 :=
\left\{ \begin{array}{l}
    0, \ \hbox{if} \  \mid A^{k,j}_{T^k_1} \mid \neq 2^{-k} \\ \\
    1, \ \hbox{if} \  \mid A^{k,j}_{T^k_1} \mid = 2^{-k},
\end{array}
\right.
\] for any $j=1, \ldots , p$. Then, we shift $A^k$ as follows

\[
\tilde{A}^k_1 (t) := \left( \tilde{A}^{k,1}_1(t) := A^{k,1} (t + T^k_1) - A^{k,1} (T^k_{\ell^{k,1}_1}) ; \ldots ; \tilde{A}^{k,p}_1(t) := A^{k,p} (t + T^k_1) - A^{k,p} (T^k_{\ell^{k,p}_1})\right),
\] for $t \geq 0$. In this case, we conclude that $\tilde{A}^k_1$ is adapted to the filtration $\{\mathcal{F}^k_{t + T^k_1}; t \geq 0\}$, the hitting time

\[
S^k_2 := \inf \{ t > 0; \parallel \tilde{A}^k_1 (t) \parallel_{\Bbb{R}^p} = 2^{-k} \}
\] is a $\{\mathcal{F}^k_{t + T^k_1}: t \geq 0\}$-stopping time and $T^k_2 = T^k_1 + S^k_2$ is a $\mathbb{F}^k$-stopping time. In the sequel, we define a family of $\mathcal{F}^k_{T^k_2}$-random variables related to the index $j$ which realizes the hitting time $T^k_2$ as follows

\[
\ell^{k,j}_2 :=
\left\{ \begin{array}{l}
    0, \ \hbox{if} \  \mid \tilde{A}^{k,j}_1 (S^k_2) \mid \neq 2^{-k} \\
    2, \ \hbox{if} \  \mid \tilde{A}^{k,j}_1 (S^k_2) \mid = 2^{-k},
\end{array}
\right.
\]
for $j=1, \ldots , p$. If we denote $S^k_0 =0$, we shift $\tilde{A}^k_1$ as follows

\[
\tilde{A}^k_2 (t) := \left( \tilde{A}^{k,1}_2(t) := \tilde{A}^{k,1}_1 (t + S^k_2) - \tilde{A}^{k,1}_1 (S^k_{\ell^{k,1}_2}); \ldots ; \tilde{A}^{k,p}_2(t) = \tilde{A}^{k,p} (t + S^k_2) - \tilde{A}^{k,p} (S^k_{\ell^{k,p}_2})\right),
\] for every $t \geq 0$. In this case, we conclude that $\tilde{A}^k_2$ is adapted to the filtration $\{\mathcal{F}^k_{t + T^k_2}; t \geq 0\}$, the hitting time

\[
S^k_3 = \inf \{ t > 0 ; \parallel \tilde{A}^k_2 (t) \parallel_{\Bbb{R}^p} = 2^{-k} \}
\] is an $\{\mathcal{F}^k_{t + T^k_2}; t \geq 0\}$-stopping time and $T^k_3 = T^k_2 + S^k_3$ is a $\mathbb{F}^k$-stopping time. By induction, we conclude that $(T^k_n)_{n\ge 0}$ is a sequence of $\mathbb{F}^k$-stopping times.
\end{proof}

With Lemma~\ref{st} at hand, we notice that the filtration $\mathbb{F}^k$ is a discrete-type filtration in the sense that

$$\mathcal{F}^k_{T^k_n}=\mathcal{F}^k_t~a.s~\text{on}~\{T^k_n \le t < T^k_{n+1} \},$$
for $k\ge 1$ and $n\ge 0$. It\^o representation theorem yields

$$
\mathbb{E}[H|\mathcal{F}_t] = \mathbb{E}[H] + \int_0^t \phi^H_u dW_u; \quad 0\le t\le T,
$$
where $\phi^H$ is a $p$-dimensional $\mathbb{F}$-predictable process such that
$$
\mathbb{E}\int_0^T\|\phi^H_t\|^2_{\mathbb{R}^p}dt<\infty.
$$
The payoff $H$ induces the $\mathbb{Q}$-square-integrable $\mathbb{F}$-martingale $X_t:=\mathbb{E}[H|\mathcal{F}_t];~0\le t\le T$. We now embed the process $X$ into the quasi left-continuous filtration $\mathbb{F}^{k}$ by means of the following operator

$$
\delta^{k}X_t:= X_0 + \sum_{m=1}^\infty \mathbb{E}\big[X_{T^k_m}|\mathcal{F}^k_{T^k_m}\big] 1\!\!1_{\{T^k_m\leq t< T^k_{m+1}\}};~0\le t\le T.
$$
Since $X$ is a $\mathbb{F}$-martingale, then the usual optional stopping theorem yields the representation
\[
\delta^{k}X_t=\mathbb{E}[X_T|\mathcal{F}^k_t] = \mathbb{E}[H|\mathcal{F}^k_t], \quad 0\le t\le T.
\]
Therefore, $\delta^{k}X$ is indeed a $\mathbb{Q}$-square-integrable $\mathbb{F}^k$-martingale and we shall write it as

\begin{align}\label{deltaXdef}
\delta^kX_t & = X_0 + \sum_{m=1}^\infty  \Delta \delta^{k} X_{T^k_m} 1\!\!1_{\{T^k_m \leq t\}}
=X_0 + \sum_{j=1}^p\sum_{n=1}^\infty  \Delta \delta^{k} X_{T^{k,j}_n} 1\!\!1_{\{T^{k,j}_n \leq t\}} \nonumber \\
& = X_0 + \sum_{j=1}^{p} \sum_{\ell=1}^\infty  \frac{\Delta \delta^{k}X_{T^{k,j}_\ell}}{\Delta A^{k,j}_{T^{k,j}_\ell}} \Delta A^{k,j}_{T^{k,j}_\ell} 1\!\!1_{\{T^{k,j}_\ell \le t\}} = X_0 + \sum_{j=1}^{p} \int_{0}^t \mathcal{D}^j\delta^{k} X_u dA^{k,j}_u,
\end{align}
where
$$
\mathcal{D}^j \delta^{k} X :=  \sum_{\ell=1}^{\infty} \frac{\Delta \delta^{k}X_{T^{k,j}_\ell}}{\Delta A^{k,j}_{T^{k,j}_\ell}} 1\!\!1_{[[T^{k,j}_\ell,T^{k,j}_\ell ]]},
$$
and the integral in \eqref{deltaXdef} is computed in the Lebesgue-Stieltjes sense.
\begin{remark}\label{delcon}
Similar to the univariate case, one can easily check that $\mathbb{F}^k\rightarrow \mathbb{F}$ weakly and since $X$ has continuous paths then $\delta^kX\rightarrow X$ uniformly in probability as $k\rightarrow \infty$. See Remark 2.1 in~\cite{LEAO_OHASHI09}.
\end{remark}
Based on the Dirac process $\mathcal{D}^j\delta^k X$, we denote

$$
\mathbb{D}^{k,j} X :=\sum_{\ell=1}^{\infty} \mathcal{D}^j_{T^{k,j}_\ell} \delta^{k} X1\!\!1_{[[T^{k,j}_{\ell}, T^{k,j}_{\ell + 1}[[},~k\ge 1, j=1,\ldots, p.
$$

In order to work with non-antecipative hedging strategies, let us now define a suitable $\mathbb{F}^k$-predictable version of $\mathbb{D}^{k,j}X$ as follows

\[
\mathbf{D}^{k,j}X:= 01\!\!1_{[[0]]} + \sum_{n=1}^{\infty}\mathbb{E}\big[\mathbb{D}^{k,j}X_{T^{k,j}_n}| \mathcal{F}^{k}_{T^{k,j}_{n-1}}\big]1\!\!1_{ ]]T^{k,j}_{n-1}, T^{k,j}_n]]}; k\ge 1, j=1,\ldots, d.
\]
One can check that $\mathbf{D}^{k,j}X$ is $\mathbb{F}^k$-predictable. See e.g~\cite{He92}, Ch.5 for details.

\

\noindent \textbf{Example}: Let $H$ be a contingent claim satisfying \textbf{(M)}. Then
for a given $j=1,\dots, p$, we have

\begin{equation}\label{exder}
\mathbf{D}^{k,j} X_{t} = \mathbb{E} \sum_{\ell=1}^{\infty} \Bigg[\frac{\mathbb{E}\big[H\big| \mathcal{F}^k_{T^{k}_\ell} \big] - \mathbb{E}\big[H\big|\mathcal{F}^k_{T^{k}_{\ell-1}}\big]}{W^{(j)}_{T^{k,j}_1} -W^{(j)}_{T^{k,j}_0} }\Bigg] 1\!\!1_{ \{ T^{k,j}_1 = T^k_{\ell}  \} }, \quad 0 < t \leq T^{k,j}_1.
\end{equation}
One should notice that~(\ref{exder}) is reminiscent from the usual delta-hedging strategy but the price is shifted on the level of the sigma-algebras jointly with the increments of the driving Brownian motion instead of the pure spot price. For instance, in the one-dimensional case $(p=d=1)$, we have

$$
\mathbf{D}^{k,1} X_{t} = \mathbb{E} \Bigg[\frac{\mathbb{E}\big[H\big| \mathcal{F}^k_{T^{k,1}_1} \big] - \mathbb{E}[H]}{W^{(1)}_{T^{k,1}_1} -W^{(1)}_{T^{k,1}_0} }\Bigg], \quad 0 < t \leq T^{k,1}_1,
$$
and hence a natural procedure to approximate pure hedging strategies is to look at $\mathbf{D}^{k,1}X_{T^{k,1}_1} / S_0\sigma_0$ at time zero. In the incomplete market case, additional randomness from e.g stochastic volatilities are encoded by $\mathbb{E}[H|\mathcal{F}^{k}_{T^{k}_1}]$ where $T^{k}_1$ is determined not only by the hitting times coming from the risky asset prices but also by possibly Brownian motion hitting times coming from stochastic volatility.

In the next sections, we will construct feasible approximations for the gain and cost processes based on the ratios~(\ref{exder}). We will see that hedging ratios of the form~(\ref{exder}) will be the key ingredient to recover the gain process in full generality.

\

\subsection{Weak approximation for the hedging process}
Based on \eqref{cost}, \eqref{pfsg} and \eqref{keyprocess}, let us denote

\begin{equation}\label{GHedge_Strategy}
\theta^H_t := \phi^{H,S}_t \left[\diag(S_t)\sigma_t \right]^{-1}\quad \mbox{and} \quad L^H_t:= \mathbb{E}[H]+ \int_0^t\phi^{H,I}_\ell dW^I_\ell;~0\le t\le T.
\end{equation}
In order to shorten notation, we do not write $(\phi^{H,\mathbb{Q},S},\phi^{H,\mathbb{Q},I})$ in \eqref{GHedge_Strategy}. The main goal of this section is the obtention of bounded variation martingale weak approximations for both the gain and cost processes, given respectively, by

$$\int_{0}^t \theta^{H}_u d S_u,\quad L^H_t;~0\le t\le T.$$
We assume the trader has some knowledge of the underlying volatility so that the obtention of $\phi^{H,S}$ will be sufficient to recover $\theta^{H}$. The typical example we have in mind are generalized F\"{o}llmer-Schweizer decompositions, locally-risk minimizing and mean variance strategies as explained in the Introduction. The scheme will be very constructive in such way that all the elements of our approximation will be amenable to a feasible numerical analysis. Under very mild integrability conditions, the weak approximations for the gain process will be translated into the physical measure.

\

\noindent \textit{The weak topology}. In order to obtain approximation results under full generality, it is important to consider a topology which is flexible to deal with nonsmooth hedging strategies $\theta^H$ for possibly non-Markovian payoffs $H$ and at the same time justifies Monte Carlo procedures. In the sequel, we make use of the weak topology $\sigma(\text{B}^p,\text{M}^q)$ of the Banach space $\text{B}^p(\mathbb{F})$ constituted by $\mathbb{F}$-optional processes $Y$ such that

$$\mathbb{E}|Y^*_T|^p < \infty,$$
where $Y^*_T:=\sup_{0\le t\le T}|Y_t|$ and $1\le p, q<\infty$ such that $\frac{1}{p} +\frac{1}{q}=1$. The subspace of the square-integrable $\mathbb{F}$-martingales will be denoted by $\text{H}^2(\mathbb{F})$. It will be also useful to work with $\sigma(\text{B}^1,\Lambda^\infty)$-topology given in~\cite{LEAO_OHASHI09}. For more details about these topologies, we refer to the works~\cite{dellacherie,dellacherie1978,LEAO_OHASHI09}. It turns out that $\sigma(\text{B}^2,\text{M}^2)$ and $\sigma(\text{B}^1,\Lambda^\infty)$ are very natural notions to deal with generic square-integrable random variables as described in~\cite{LEAO_OHASHI09}.

In the sequel, we recall the following notion of covariation introduced in~\cite{LEAO_OHASHI09}.

\begin{definition}\label{candelta}
Let $\{Y^k;k\ge 1\}$ be a sequence of square-integrable $\mathbb{F}^k$-martingales. We say that $\{Y^k; k\ge 1\}$ has $\delta$-covariation w.r.t jth component of $A^k$ if the limit

$$\lim_{k\rightarrow \infty} [Y^k, A^{k,j}]_t$$
exists weakly in $L^1(\mathbb{Q})$ for every $t\in [0,T]$.
\end{definition}

\begin{lemma}\label{intres}
Let $\Big\{ Y^{k,j} = \int_0^\cdot H^{k,j}_sdA^{k,j}; k\ge 1, j=1,\ldots, p  \Big\}$ be a sequence of stochastic integrals and $Y^k:=\sum_{j=1}^pY^{k,j}$. Assume that
$$
\sup_{k\ge 1} \mathbb{E}[Y^{k}, Y^k]_T < \infty.
$$
Then $Y^j:=\lim_{k\rightarrow \infty}Y^{k,j}$ exists weakly in $\text{B}^2(\mathbb{F})$ for each $j=1,\ldots,p$ with $Y^j\in\text{H}^2(\mathbb{F})$~if, and only if, $\{Y^{k};k\ge 1\}$~admits $\delta$-covariation w.r.t~jth component of $A^k$. In this case,

$$\lim_{k\rightarrow \infty}[Y^{k},A^{k,j}]_t =\lim_{k\rightarrow \infty}[Y^{k,j},A^{k,j}]_t= [Y^j,W^{(j)}]_t\quad\text{weakly in}~L^1;~t\in [0,T]$$
for $j=1,\ldots, p$.
\end{lemma}
\begin{proof}
The proof follows easily from the arguments given in the proof of Prop. 3.2 in~\cite{LEAO_OHASHI09} by using the fact that $\{W^{(j)}; 1\le j\le p\}$ is an independent family of Brownian motions, so we omit the details.
\end{proof}
In the sequel, we present a key asymptotic result for the numerical algorithm of this article.

\begin{theorem}\label{deltaj}
Let $H$ be a $\mathbb{Q}$-square integrable contingent claim satisfying~\textbf{(M)}. Then


\begin{equation}\label{convcomp}
\lim_{k\rightarrow \infty}\sum_{j=1}^d \int_0^\cdot\mathbf{D}_s^{k,j} XdA^{k,j}_s= \sum_{j=1}^d\int_0^\cdot \phi^{H,j}_udW^{(j)}_u = \int_0^\cdot \theta^H_udS_u,
\end{equation}
and
\begin{equation}\label{wlim1}
L^{H}=\lim_{k\rightarrow \infty}\sum_{j=d+1}^{p}\int_0^\cdot\mathbf{D}_s^{k,j} XdA^{k,j}_s
\end{equation}
weakly in $B^2(\mathbb{F})$. In particular,

\begin{equation}\label{optionalpr}
\lim_{k\rightarrow\infty} \mathbf{D}^{k,j} X=\phi^{H,j},
\end{equation}
weakly in~$L^2(Leb\times\mathbb{Q})$ for each $j=1,\ldots,p.$
\end{theorem}
\begin{proof}
We divide the proof into three steps. Throughout this proof $C$ is a generic constant which may defer from line to line.

\

\noindent \textbf{STEP1.}~ We claim that

\begin{equation}\label{f1}
\lim_{k\rightarrow \infty}\int_0^\cdot\mathbb{D}^{k,j}X_sdA^{k,j}_s = \int_0^\cdot\phi^{H,j}_udW^{(j)}_u\quad\text{weakly in}~\text{B}^2(\mathbb{F})
\end{equation}
for each $j=1,\ldots, p$. In order to prove~(\ref{f1}), we begin by noticing that Lemma~\ref{st} states that the elements of $\mathcal{T}^k$ are $\mathbb{F}$-stopping times. By assumption, $X$ is $\mathbb{Q}$-square integrable martingale and hence one may use similar arguments given in the proof of Lemma 3.1 in~\cite{LEAO_OHASHI09} to safely state that the following estimate holds

\begin{equation}\label{menergy}
\sup_{k\ge 1}\mathbb{E}[\delta^kX,\delta^kX]_T = \sup_{k\ge 1}\mathbb{E}\sum_{j=1}^p\int_0^T|\mathbb{D}^{k,j}X_s|^2d[A^{k,j}, A^{k,j}]_s   \le \sup_{k\ge 1}\mathbb{E}\sum_{m=1}^\infty (X_{T^k_m}-X_{T^k_{m-1}})^2 1\!\!1_{\{T^k_m \leq T\}} <  \infty.
\end{equation}

\

Now, we notice that the sequence $\mathbb{F}^k$ converges weakly to $\mathbb{F}$, $X$ is continuous and therefore $\delta^kX \rightarrow X$ uniformly in probability (see Remark~\ref{delcon}). Since $X\in \text{B}^2(\mathbb{F})$, then a simple application of Burkh\"{o}lder inequality allows us to state that $\delta^kX$ converges strongly in $\text{B}^1(\mathbb{F})$ and a routine argument based on the definition of the $\text{B}^2$-weak topology yields

\begin{equation}\label{f2}
\lim_{k\rightarrow \infty}\delta^k X=X~\text{weakly in}~\text{B}^2(\mathbb{F}).
\end{equation}
Now under~(\ref{f2}) and~(\ref{menergy}), we shall prove in the same way as in Prop.3.2 in~\cite{LEAO_OHASHI09} that


\begin{equation}\label{dx}
\lim_{k\rightarrow \infty}[\delta^kX, A^{k,j}]_t = [X, W^{(j)}]_t =\int_0^t \phi^{H,j}_u du;0\le t\le T,
\end{equation}
holds weakly in $L^1$ for each $t\in [0,T]$ and $j=1,\ldots, p$ due to the pairwise independence of $\{W^{(j)}; 1\le j\le p\}$. Summing up~(\ref{menergy}) and~(\ref{dx}), we shall apply Lemma~\ref{intres} to get~(\ref{f1}).

\

\noindent \textbf{\noindent \textbf{STEP 2.}} In the sequel, let $(\cdot)^{o,k}$ and $(\cdot)^{p,k}$ be the optional and predictable projections w.r.t $\mathbb{F}^{k}$, respectively. Let us consider the $\mathbb{F}^k$-martingales given by
$$M^k_t:=\sum_{j=1}^p M^{k,j}_t;~0\le t\le T,$$
where
$$M^{k,j}_t := \int_0^t\mathbf{D}^{k,j}X_sdA^{k,j}_s;~0\le t\le T,~j=1,\ldots, p.$$
We claim that $\sup_{k\ge 1}\mathbb{E}[M^k,M^k]_T < \infty$. One can check that $\mathbf{D}^{k,j}X_{T^{k,j}_n} = \Big(\mathbb{D}^{k,j}X\Big)^{p,k}_{T^{k,j}_n}$ a.s~ for each $n,k\ge 1$ and $j=1\ldots, p$ (see e.g~chap.5, section 5 in~\cite{He92}). Moreover, by the very definition

\begin{equation}\label{deltaA}
\{(t,\omega)\in [0,T]\times \Omega; \Delta [A^{k,j},A^{k,j}]_t(\omega) \neq 0\} = \bigcup_{n=1}^{\infty}[[T^{k,j}_n, T^{k,j}_n]].
\end{equation}
Therefore, Jensen inequality yields

\begin{eqnarray}
\nonumber \mathbb{E}[M^k,M^k]_T &=&\mathbb{E} \sum_{j=1}^p\int_0^T|\mathbf{D}^{k,j}X_s|^2d[A^{k,j},A^{k,j}]_s\\
\nonumber & &\\
\nonumber&=&\mathbb{E} \sum_{j=1}^p\int_0^T\Big|\Big(\mathbb{D}^{k,j}X\Big)_s^{p,k}\Big|^2d[A^{k,j},A^{k,j}]_s\\
\nonumber& &\\
\nonumber&\le&  \mathbb{E}\sum_{j=1}^p\int_0^T \Big( (\mathbb{D}^{k,j}X_s)^2\Big)^{p,k}_s d[ A^{k,j}, A^{k,j}]_s\\
\nonumber& &\\
\label{f3}&=& \sum_{j=1}^p\mathbb{E}\sum_{n=1}^\infty\mathbb{E}\big[(\mathbb{D}^{k,j}X_{T^{k,j}_{n}})^2 | \mathcal{F}^{k}_{T^{k,j}_{n-1}} \big]2^{-2k}1\!\!1_{ \{T_{n}^{k,j} \leq T  \} }:=J^{k},
\end{eqnarray}
where in~(\ref{f3}) we have used~(\ref{deltaA}) and the fact that $\Big((\mathbb{D}^{k,j}X)^2\Big)^{p,k}_{T^{k,j}_{n}} = \mathbb{E}\big[(\mathbb{D}^{k,j}X_{T^{k,j}_{n}})^2 | \mathcal{F}^{k}_{T^{k,j}_{n-1}} \big]$ a.s for each $n,k\ge 1$ and $j=1\ldots, p$.  We shall write~$J^{k}$ in a slightly different manner as follows

\begin{equation}\label{trick}
J^{k} = \sum_{j=1}^p\mathbb{E}\sum_{n=1}^\infty\mathbb{E}\Big[(\mathbb{D}^{k,j}X_{T^{k,j}_{n}})^2 | \mathcal{F}^{k}_{T^{k,j}_{n-1}} \Big]2^{-2k}1\!\!1_{ \{T_{n-1}^{k,j} \leq T  \} }
\end{equation}

$$ - \sum_{j=1}^p\mathbb{E}\Big[\mathbb{E}\big[(\mathbb{D}^{k,j}X_{T^{k,j}_{q}})^2
|\mathcal{F}^{k}_{T^{k,j}_{q-1}}\big]\Big]2^{-2k}1\!\!1_{ \{T_{q-1}^{k,j} \leq T < T^{k,j}_q  \} }$$

$$
 = \sum_{j=1}^p\mathbb{E}\sum_{n=1}^\infty|\Delta \delta^k X_{T^{k,j}_{n}}|^2 1\!\!1_{ \{T_{n-1}^{k,j} \leq T  \} }
$$

$$ - \sum_{j=1}^p\mathbb{E}\Big[\mathbb{E}\big[(\Delta\delta^k X_{T^{k,j}_{q}})^2
|\mathcal{F}^{k}_{T^{k,j}_{q-1}}\big]\Big]1\!\!1_{ \{T_{q-1}^{k,j} \leq T < T^{k,j}_q  \} }.$$
The above identities, the estimates~(\ref{menergy}),~(\ref{f3}) and Remark~\ref{delcon} yield

\begin{equation}\label{g1}
\limsup_{k\rightarrow \infty}\mathbb{E}[M^k,M^k]_T \le \limsup_{k\rightarrow\infty}J^k<\infty.
\end{equation}

\

\noindent \textbf{STEP 3}. We claim that for a given $g\in L^\infty$, $t\in [0,T]$ and $j=1\ldots, p$ we have

\begin{equation}\label{f4}
\lim_{k\rightarrow \infty}\mathbb{E}g[M^k-\delta^kX, A^{k,j}]_t = 0.
\end{equation}
By using the fact that $\mathbb{D}^{k,j}X$ is $\mathbb{F}^k$-optional and $\mathbf{D}^{k,j}X$ is $\mathbb{F}^k$-predictable, we shall use duality of the $\mathbb{F}^{k}$-optional projection to write

$$\mathbb{E}g[M^k-\delta^kX, A^{k,j}]_t = \mathbb{E}\int_0^t (g)^{o,k}_s \Big( \mathbf{D}^{k,j}X_s  - \mathbb{D}^{k,j}X_s\Big)d[A^{k,j},A^{k,j}]_s.$$
In order to prove~(\ref{f4}), let us check that

\begin{equation}\label{f5}
\lim_{k\rightarrow\infty}\mathbb{E}\int_0^t(g)^{p,k}_s\Big(  \mathbf{D}^{k,j}X_s  - \mathbb{D}^{k,j}X_s  \Big)d[A^{k,j},A^{k,j}]_s = 0,
\end{equation}
and

\begin{equation}\label{f6}
\lim_{k\rightarrow\infty}\mathbb{E}\int_0^t\big((g)^{o,k}_s - (g)^{p,k}_s\big)\Big(  \mathbf{D}^{k,j}X_s  - \mathbb{D}^{k,j}X_s  \Big)d[A^{k,j},A^{k,j}]_s = 0.
\end{equation}
The same trick we did in~(\ref{trick}) together with (\ref{deltaA}) yield

$$\mathbb{E}\int_0^t(g)^{p,k}_s\Big(  \mathbf{D}^{k,j}X_s  - \mathbb{D}^{k,j}X_s  \Big)d[A^{k,j},A^{k,j}]_s =
\mathbb{E}\Big[(g)^{p,k}_{T^{k,j}_q}\mathbb{D}^{k,j}X_{T^{k,j}_q}\Big]2^{-2k}1\!\!1_{ \{T_{q-1}^{k,j} \leq t < T^{k,j}_q  \}} $$

$$- \mathbb{E}  \Big[ \mathbb{E}\big[(g)^{p,k}_{T^{k,j}_q}\mathbb{D}^{k,j}X_{T^{k,j}_{q}} | \mathcal{F}^{k}_{T^{k,j}_{q-1}} \big] \Big]2^{-2k}1\!\!1_{ \{T_{q-1}^{k,j} \leq t < T^{k,j}_q  \} }\rightarrow0
$$
as $k\rightarrow \infty$ because $X$ has continuous paths~(see Remark~\ref{delcon}). This proves~(\ref{f5}). Now, in order to shorten notation let us denote $I^{k,j}$ by the expectation in~(\ref{f6}). Cauchy-Schwartz and Burkholder-Davis-Gundy inequalities jointly with~(\ref{g1}) and~(\ref{menergy}) yield

$$|I^{k,j}|\le \mathbb{E}^{1/2}\sup_{0< \ell\le T }|(g)^{o,k}_\ell  - (g)^{p,k}_\ell|^2$$
$$\times \Big\{ \mathbb{E}^{1/2}\int_0^T |  \mathbf{D}^{k,j}X_s  - \mathbb{D}^{k,j}X_s |^2 d[A^{k,j},A^{k,j}]_s \times \mathbb{E}^{1/2}[A^{k,j},A^{k,j}]_T     \Big\}^{1/2}$$

\begin{equation}\label{f7}
\le C \mathbb{E}^{1/2}\sup_{n\ge 1} |\mathbb{E}[g|\mathcal{F}^{k}_{T^{k,j}_n} ] - \mathbb{E}[g|\mathcal{F}^{k}_{T^{k,j}_{n-1}}]|^2 1\!\!1_{ \{T_{n}^{k,j} \leq T  \} }.
\end{equation}

We shall proceed similar to Lemma 4.1 in~\cite{LEAO_OHASHI09} to safely state that $\mathbb{E}^{1/2}\sup_{n\ge 1} |\mathbb{E}[g|\mathcal{F}^{k}_{T^{k,j}_n} ] - \mathbb{E}[g|\mathcal{F}^{k}_{T^{k,j}_{n-1}}]|^2 1\!\!1_{ \{T_{n}^{k,j} \leq T  \} }\rightarrow 0$ as $k\rightarrow\infty$ and from~(\ref{f7}) we conclude that (\ref{f6}) holds. Summing up Steps 1, 2 and 3, we shall use Lemma~\ref{intres} to conclude that~(\ref{convcomp}) and (\ref{wlim1}) hold true. It remains to show~(\ref{optionalpr}) but this is a straightforward consequence of (\ref{f4}) together with a similar argument given in the proof of Theorem 4.1 and Remark 4.2 in~\cite{LEAO_OHASHI09}, so we omit the details. This concludes the proof of the theorem.
\end{proof}



Stronger convergence results can be obtained under rather weak integrability and path smoothness assumptions for representations $(\phi^1,\ldots, \phi^p)$. We refer the reader to the Appendix for further details.

\section{Weak dynamic hedging}\label{wdhsection}
In this section, we apply Theorem~\ref{deltaj} for the formulation of a dynamic hedging strategy starting with a given GKW decomposition

\begin{equation}\label{gfs}
H =\mathbb{E}[H] + \int_0^T\theta^H_t dS_t + L^H_T,
\end{equation}
where $H$ is a $\mathbb{Q}$-square integrable European-type option satisfying \textbf{(M)} for a given $\mathbb{Q}\in \mathcal{M}^e$. The typical examples we have in mind are quadratic hedging strategies w.r.t a fully path-dependent option. We recall that when $\mathbb{Q}$ is the minimal martingale measure then~(\ref{gfs}) is the generalized F\"{o}llmer-Schweizer decomposition so that under some $\mathbb{P}$-square integrability conditions on the components of~(\ref{gfs}), $\theta^H$ is the locally risk minimizing hedging strategy~(see~e.g~\cite{Heath},~\cite{Schweizer1995}). In fact, GKW and F\"{o}llmer-Schweizer decompositions are essentially equivalent for the market model assumed in Section~\ref{capitulo:modelo}. We recall that decomposition~(\ref{gfs}) is not sufficient to fully describe mean variance hedging strategies but the additional component rests on the fundamental representation equations as described in Introduction. See also expression~(\ref{hobeq}) in Section~\ref{capitulo:resultados}.

For simplicity of exposition, we consider a financial market $(\Omega,\textbf{F},\mathbb{P})$ driven by a two-dimensional Brownian motion $B$ and a one-dimensional risky asset price process $S$ as described in Section~\ref{capitulo:modelo}. We stress that all results in this section hold for a general multidimensional setting with the obvious modifications.

In the sequel, we denote

$$\theta^{k,H}: = \sum_{n=1}^\infty \frac{\mathbf{D}^{k,1}X_{T^{k,1}_n}} {\sigma_{T^{k,1}_{n-1}}S_{T^{k,1}_{n-1}}}
1\!\!1_{[[T^{k,1}_{n-1}, T^{k,1}_n [[}$$
where $\mathbf{D}^{k,1}X_{T^{k,1}_n}=\mathbb{E}\Big[\mathbb{D}^{k,1}X_{T^{k,1}_n}| \mathcal{F}^{k}_{T^{k,1}_{n-1}}\Big]$ for $k,n\ge 1$.

\begin{corollary}\label{changeP}
For a given $\mathbb{Q}\in\mathcal{M}^e$, let $H$ be a $\mathbb{Q}$-square integrable claim satisfying~\textbf{(M)}. Let

$$
H = \mathbb{E}[H] + \int_0^T\theta^H_tdS_t + L^H_T
$$
be the correspondent GKW decomposition under $\mathbb{Q}$. If $\frac{d\mathbb{P}}{d\mathbb{Q}} \in L^{1}(\mathbb{P})$ and

\begin{equation}\label{sup1}
\mathbb{E}_{\mathbb{P}}\sup_{0\le t\le T}\Big|\int_0^t \theta^H_udS_u\Big| <\infty,
\end{equation}
then

$$
\sum_{n=1}^\infty \theta^{k,H}_{T^{k,1}_{n-1}}(S_{T^{k,1}_n} - S_{T^{k,1}_{n-1}})1\!\!1_{\{T^{k,1}_n\le \cdot\}}\rightarrow \int_0^\cdot\theta^H_tdS_t\quad\text{as}~k\rightarrow \infty,
$$
in the $\sigma(\text{B}^1,\Lambda^\infty)$-topology under $\mathbb{P}$.
\end{corollary}
\begin{proof}
We have $\mathbb{E}|\frac{d\mathbb{P}}{d\mathbb{Q}}|^2 =
\mathbb{E}_{\mathbb{P}}|\frac{d\mathbb{P}}{d\mathbb{Q}}|^2\frac{d\mathbb{Q}}{d\mathbb{P}} = \mathbb{E}_\mathbb{P}\frac{d\mathbb{P}}{d\mathbb{Q}} < \infty$. To shorten notation, let $Y^k_t:=\int_0^t\mathbf{D}_s^{k,1} XdA^{k,1}_s$ and $Y_t:=\int_0^t\theta^H_\ell dS_\ell$ for $0\le t\le T.$ Let $G$ be an arbitrary $\mathbb{F}$-stopping time bounded by $T$ and let $g\in L^{\infty}(\mathbb{P})$ be an essentially $\mathbb{P}$-bounded random variable and $\mathcal{F}_G$-measurable. Let $J\in \text{M}^2$ be a continuous linear functional given by the purely discontinuous $\mathbb{F}$-optional bounded variation process

$$J_t:=g\mathbb{E}\Big[\frac{d\mathbb{P}}{d\mathbb{Q}}\big|\mathcal{F}_G\Big]1\!\!1_{ \{G\le t \} };0\le t\le T,$$
where the duality action $\big(\cdot~,~\cdot\big)$ is given by $ \big( J, N\big) = \mathbb{E}\int_0^T N_sdJ_s; N\in \text{B}^2(\mathbb{F})$. See~\cite{LEAO_OHASHI09} for more details. Then Theorem~\ref{deltaj} and the fact $\frac{d\mathbb{P}}{d\mathbb{Q}} \in L^2(\mathbb{Q})$ yield

$$\mathbb{E}_\mathbb{P} gY^k_G  =\mathbb{E} Y^k_Gg\frac{d\mathbb{P}}{d\mathbb{Q}} = \big(J, Y^k\big) \rightarrow \big(J, Y \big)=
\mathbb{E}Y_Gg\frac{d\mathbb{P}}{d\mathbb{Q}} = \mathbb{E}_{\mathbb{P}}gY_G$$
as $k\rightarrow \infty$. By the very definition,

\begin{eqnarray}\label{deff}
\int_0^t\mathbf{D}_s^{k,1} XdA^{k,1}_s &=& \sum_{n=1}^\infty
\mathbb{E}\Big[\mathbb{D}^{k,1}X_{T^{k,1}_n}\big|\mathcal{F}^k_{T^{k,1}_{n-1}}\Big]\Delta A^{k,1}_{T^{k,1}_n} 1\!\!1_{\{T^{k,1}_n\le t\}}\\
\nonumber& &\\
\nonumber &=& \sum_{n=1}^\infty\theta^{k,H}_{T^{k,1}_{n-1}} \sigma_{T^{k,1}_{n-1}}S_{T^{k,1}_{n-1}} (W^{(1)}_{T^{k,1}_{n}} -W^{(1)}_{T^{k,1}_{n-1}})1\!\!1_{\{T^{k,1}_n\le t\}} \\
\nonumber& &\\
\nonumber&=& \sum_{n=1}^\infty \theta^{k,H}_{T^{k,1}_{n-1}}(S_{T^{k,1}_n} - S_{T^{k,1}_{n-1}})1\!\!1_{\{T^{k,1}_n\le t\}};~0\le t\le T.
\end{eqnarray}
Then from the definition of the $\sigma(\text{B}^1,\Lambda^\infty)$-topology based on the physical measure $\mathbb{P}$, we shall conclude the proof.
\end{proof}

\begin{remark}\label{L1h}
Corollary~\ref{changeP} provides a non-antecipative Riemman-sum approximation for the gain process $\int_0^\cdot\theta^H_tdS_t$ in a multi-dimensional filtration setting where none path regularity of the pure hedging strategy $\theta^H$ is imposed. The price we pay is a weak-type convergence instead of uniform convergence in probability. However, from the financial point of view this type of convergence is sufficient for the implementation of Monte Carlo methods in hedging. More importantly, we will see that $\theta^{k,H}$ can be fairly simulated and hence the resulting Monte Carlo hedging strategy can be calibrated from market data.
\end{remark}

\begin{remark}
If one is interested only at convergence at the terminal time $0< T <\infty$, then assumption~(\ref{sup1}) can be weakened to $\mathbb{E}_{\mathbb{P}}|\int_0^T\theta^H_tdS_t|< \infty$. Assumption $\mathbb{E}_{\mathbb{P}}\frac{d\mathbb{P}}{d\mathbb{Q}} < \infty$ is essential to change the $\mathbb{Q}$-convergence into the physical measure $\mathbb{P}$. One should notice that the associated density process is no longer a $\mathbb{P}$-local-martingale and in general such integrability assumption must be checked case by case. Such assumption holds locally for every underlying It\^o risky asset price process. Our numerical results suggest that this property behaves well for a variety of spot price models.
\end{remark}
Of course, in practice both the spot prices and trading dates are not observable at the stopping times so we need to translate our results to a given deterministic set of rebalancing hedging dates.

\subsection{Hedging Strategies}\label{dished}
In this section, we provide a dynamic hedging strategy based on a refined set of hedging dates $\Pi:=0=s_0 < \ldots < s_{p-1} < s_{p} = T$. For this, we need to introduce some objects. For a given $s_i\in \Pi$, we set $W^{(j)}_{s_i,t}:= W^{(j)}_{s_i+t} - W^{(j)}_{s_i};~0\le t\le T-s_i$ for $j=1,2$. Of course, by the strong Markov property of the Brownian motion, we know that $W^{(j)}_{s_i,\cdot}$ is an $(\mathcal{F}^j_{s_i,t})_{0\le t\le  T-s_i}$-Brownian motion for each $j=1,2$ and independent from $\mathcal{F}^j_{s_i}$, where $\mathcal{F}^j_{s_i,t}:=\mathcal{F}^j_{s_i+t}$ for $0\le t\le T-s_i$. Similar to Section~\ref{bo}, we set $T^{k,1}_{s_i,0}:=0$ and

$$T^{k,j}_{s_i,n}:= \inf\{ t> T^{k,j}_{s_i,n-1}; |W^{(j)}_{s_i,t} -  W^{(j)}_{s_i,T^{k,j}_{s_i,n-1}} | = 2^{-k} \}; n\ge 1, j=1,2.$$
For a given $k\ge 1$ and $j=1,2$, we define $\mathcal{H}^{k,j}_{s_i,n}$ as the sigma-algebra generated by $\{T^{k,j}_{s_i,\ell};1\le \ell \le n\}$ and $W^{(j)}_{s_i,T^{k,j}_{s_i,\ell}} -W^{(j)}_{s_i,T^{k,j}_{s_i,\ell-1}};1\le\ell \le n$. We then define the following discrete jumping filtration

$$\mathcal{F}^{k,j}_{s_i,t}: =\mathcal{H}^{k,j}_{s_i,n}~a.s~\text{on}~\{T^{k,j}_{s_i,n} \le t < T^{k,j}_{s_i,n+1} \}. $$
In order to deal with fully path dependent options, it is convenient to introduce the following augmented filtration

$$\mathcal{G}^{k,j}_{s_i,t}: = \mathcal{F}^{j}_{s_i}\vee \mathcal{F}^{k,j}_{s_i,t};~0\le t\le T - s_i,$$
for $j=1,2$. The bidimensional information flows are defined by $\mathcal{F}_{s_i,t}: = \mathcal{F}^1_{s_i,t}\otimes \mathcal{F}^{2}_{s_i,t}$ and $\mathcal{G}^{k}_{s_i,t}: = \mathcal{G}^{k,1}_{s_i,t}\otimes \mathcal{G}^{k,2}_{s_i,t}$ for $0\le t \le T-s_i$. We set $\mathbb{G}^{k}_{s_i}: = \{\mathcal{G}^{k}_{s_i,t}; 0\le t\le T-s_i\}$. We shall assume that they satisfy the usual conditions. The piecewise constant martingale projection $A^{k,j}_{s_i}$ based on $W^{(j)}_{s_i}$ is given by

$$A^{k,j}_{s_i,t}:=\mathbb{E}[W^{(j)}_{s_i, T-s_i}| \mathcal{G}^{k,j}_{s_i,t}];0\le t\le T-s_i.$$
We set $\{T^{k}_{s_i,n}; n\ge 0\}$ as the order statistic generated by the stopping times $\{T^{k,j}_{s_i,n}; j=1,2,~n\ge 0\}$ similar to~(\ref{difst}).

If $H\in L^2(\mathbb{Q})$ and $X_t = \mathbb{E}[H|\mathcal{F}_t];~0\le t\le  T$, then we define

$$\delta^k_{s_i}X_t: = \mathbb{E}[H| \mathcal{G}^k_{s_i,t}];~0\le t\le T-s_i,$$
so that the related derivative operators are given by

$$\mathbb{D}^{k,j}_{s_i}X:= \sum_{n=1}^\infty \mathcal{D}^j_{T^{k,j}_{s_i,n}}\delta^k_{s_i}X
1\!\!1_{[[T^{k,j}_{s_i,n}, T^{k,j}_{s_i,n + 1}[[},$$
where

$$\mathcal{D}^j\delta^k_{s_i} X := \sum_{n=1}^\infty \frac{\Delta \delta^k_{s_i} X_{T^{k,j}_{s_i,n}}}{\Delta A^{k,j}_{T^{k,j}_{s_i,n}}}1\!\!1_{[[ T^{k,j}_{s_i,n},T^{k,j}_{s_i,n} ]]};~j=1,2, k\ge 1.$$

An $\mathbb{G}^k_{s_i}$-predictable version of $\mathbb{D}^{k,j}_{s_i}X$ is given by

$$\mathbf{D}^{k,j}_{s_i}X:= 01\!\!1_{[[0]]} + \sum_{n=1}^\infty \mathbb{E}\big[
\mathbb{D}^{k,j}_{s_i}X_{T^{k,j}_{s_i,n}}|\mathcal{G}^k_{s_i,T^{k,j}_{s_i,n-1}}\big]
1\!\!1_{]]T^{k,j}_{s_i,n-1}, T^{k,j}_{s_i,n}]]};j=1,2.$$
\noindent In the sequel, we denote

\begin{equation}\label{truestrat}
\theta^{k,H}_{s_i}: = \sum_{n=1}^\infty \frac{\mathbf{D}^{k,1}_{s_i}X_{T^{k,1}_{s_i,n}}} {
\sigma_{s_i,T^{k,1}_{s_i,n-1}}S_{s_i,T^{k,1}_{s_i,n-1}}}
1\!\!1_{[[T^{k,1}_{s_i,n-1}, T^{k,j}_{s_i,n}[[};~s_i \in \Pi,
\end{equation}
where $\sigma_{s_i,\cdot}$ is the volatility process driven by the shifted filtration $\{\mathcal{F}_{s_i, t}; 0\le t\le T-s_i\}$ and $S_{s_i,\cdot}$ is the risky asset price process driven by the shifted Brownian $W^{(1)}_{s_i}$.

We are now able to present the main result of this section.

\begin{corollary}\label{shiftTh}
For a given $\mathbb{Q}\in\mathcal{M}^e$, let $H$ be a $\mathbb{Q}$-square integrable claim satisfying~\textbf{(M)}. Let

$$
H = \mathbb{E}[H] + \int_0^T\theta^H_tdS_t + L^H_T
$$
be the correspondent GKW decomposition under $\mathbb{Q}$. If $\frac{d\mathbb{P}}{d\mathbb{Q}} \in L^{1}(\mathbb{P})$ and

$$
\mathbb{E}_{\mathbb{P}}\Big|\int_0^T \theta^H_udS_u\Big| <\infty,
$$
then for any set of trading dates $\Pi=\{(s_i)_{i=0}^{p}\}$, we have

\begin{equation}\label{gaink}
\lim_{k\rightarrow \infty}\sum_{s_i\in \Pi}\sum_{n=1}^\infty\theta^{k,H}_{s_i,T^{k,1}_{s_i,n-1}}
\big(S_{s_i,T^{k,1}_{s_i,n}} -S_{s_i,T^{k,1}_{s_i,n-1}} \big)1\!\!1_{\{T^{k,1}_{s_i,n} \le s_{i+1} - s_i \}} = \int_0^T\theta^H_tdS_t
\end{equation}
weakly in $L^1$ under $\mathbb{P}$.
\end{corollary}
\begin{proof}
Let $\Pi=\{(s_i)_{i=0}^{p}\}$ be any set of trading dates. To shorten notation, let us define

\begin{equation}\label{wealth}
R(\theta^{k,H},\Pi,k):=\sum_{s_i\in \Pi}\sum_{n=1}^\infty\theta^{k,H}_{s_i,T^{k,1}_{s_i,n-1}}
\big(S_{s_i,T^{k,1}_{s_i,n}} -S_{s_i,T^{k,1}_{s_i,n-1}} \big)1\!\!1_{\{T^{k,1}_{s_i,n} \le s_{i+1} - s_i \}}
\end{equation}
for $k\ge 1$and $\Pi$. At first, we recall that $\{T^{k,1}_{s_i,n} - T^{k,1}_{s_i,n-1}; n\ge 1, s_i\in \Pi\}$ is an i.i.d sequence with absolutely continuous distribution. In this case, the probability of the set $\{T^{k,1}_{s_i,n} \le s_{i+1} - s_i \} $ is always strictly positive for every $\Pi$ and $k,n \ge 1$. Hence, $R(\theta^{k,H},\Pi,k)$ is a non-degenerate subset of random variables. By making a change of variable on the It\^o integral, we shall write

$$\int_0^T\theta^H_tdS_t = \int_0^T\phi^{H,1}_tdW^{(1)}_t = \sum_{s_i\in \Pi}\int_{s_i}^{s_{i+1}}\phi^{H,1}_tdW^{(1)}_t
 = $$
\begin{equation}\label{spp}
\sum_{s_i\in \Pi}\int_0^{s_{i+1}-s_i}\phi^{H,1}_{s_i + t}dW^{(1)}_{s_i,t}.
\end{equation}

Let us fix $\mathbb{Q}\in \mathcal{M}^e$. By the very definition,

$$R(\theta^{k,H},\Pi,k) = \sum_{s_i\in \Pi}\int_0^{s_{i+1}-s_i}\mathbf{D}^{k,1}_{s_i}X_\ell dA^{k,1}_{s_i,\ell}\quad \text{under}~\mathbb{Q}$$
Now we notice that Theorem~\ref{deltaj} holds for the two-dimensional Brownian motion $\big(W^{(1)}_{s_i}, W^{(2)}_{s_i}\big)$, for each $s_i\in \Pi$ with the discretization of the Brownian motion given by $A^{k,1}_{s_i}$. Moreover, using the fact that $\mathbb{E}|\frac{d\mathbb{P}}{d\mathbb{Q}}|^2 <\infty$ and repeating the argument given by~(\ref{deff}) restricted to the interval $[s_i, s_{i+1})$, we have

\begin{eqnarray}
\nonumber\lim_{k\rightarrow \infty} R(\theta^{k,H}, \Pi, k) &=& \sum_{s_i\in \Pi}\lim_{k\rightarrow \infty}\int_0^{s_{i+1}-s_i}\mathbf{D}^{k,1}_{s_i}X_\ell dA^{k,1}_{s_i,\ell} \\
\label{iterl1}& = & \int_0^T\theta^H_tdS_t,
\end{eqnarray}
weakly in $L^1(\mathbb{P})$ for each $\Pi$. This concludes the proof.

\end{proof}

\begin{remark}\label{fundremark}
In practice, one may approximate the gain process by a non-antecipative strategy as follows: Let $\Pi$ be a given set of trading dates on the interval $[0,T]$ so that $|\Pi|=\max_{0\le i\le p}|s_{i}-s_{i-1}|$ is small. We take a large $k$ and we perform a non-antecipative buy and hold-type strategy among the trading dates $[s_i, s_{i+1});s_i\in \Pi$ in the full approximation~(\ref{wealth}) which results

\begin{equation}\label{bhs}
\sum_{s_i\in \Pi}\theta^{k,H}_{s_i,0}
\big(S_{s_i,s_{i+1}-s_i} - S_{s_i,0} \big)\quad \text{where}\quad\theta^{k,H}_{s_i,0} = \frac{\mathbb{E}\Big[\mathbb{D}^{k,1}_{s_i}X_{T^{k,1}_{s_i,1}}\big|\mathcal{F}_{s_i}\Big]}{\sigma_{s_i,0}S_{s_i,0}};s_i\in \Pi.
\end{equation}
Convergence~(\ref{gaink}) implies that the approximation~(\ref{bhs}) results in unavoidable hedging errors w.r.t the gain process due to the discretization of the dynamic hedging, but we do not expect large hedging errors provided $k$ is large and $|\Pi|$ small. Hedging errors arising from discrete hedging in complete markets are widely studied in literature. We do not know optimal rebalancing dates in this incomplete market setting, but simulation results presented in Section~\ref{capitulo:resultados} suggest that homogeneous hedging dates work very well for a variety of models with and without stochastic volatility. A more detailed study is needed in order to get more precise relations between $\Pi$ and the stopping times, a topic which will be further explored in a forthcoming paper.
\end{remark}

Let us now briefly explain how the results of this section can be applied to well-known quadratic hedging methodologies.

\

\noindent \textbf{Generalized F\"{o}llmer-Schweizer}: If one takes the minimal martingale measure $\hat{\mathbb{P}}$, then $L^{H}$ in~(\ref{gfs}) is a $\mathbb{P}$-local martingale and orthogonal to the martingale component of $S$. Due this orthogonality and the zero mean behavior of the cost $L^H$, it is still reasonable to work with generalized F\"{o}llmer-Schweizer decompositions under $\mathbb{P}$ without knowing a priori the existence of locally-risk minimizing hedging strategies.

\

\noindent \textbf{Local Risk Minimization}:
One should notice that if $\int\theta^H dS \in \text{B}^2(\mathbb{F})$, $L^H\in \text{B}^2(\mathbb{F})$ under $\mathbb{P}$ and $\frac{d\hat{\mathbb{P}}}{d\mathbb{P}}\in L^2(\mathbb{P})$, then $\theta^H$ is the locally risk minimizing trading strategy and~(\ref{gfs}) is the F\"{o}llmer-Schweizer decomposition under $\mathbb{P}$.

\

\noindent \textbf{Mean Variance hedging}: If one takes $\tilde{\mathbb{P}}$, then the mean variance hedging strategy is not completely determined by the GKW decomposition under $\tilde{\mathbb{P}}$. Nevertheless, Corollary~\ref{shiftTh} still can be used to approximate  the optimal hedging strategy by computing the density process $\tilde{Z}$ based on the so-called fundamental equations derived by Hobson~\cite{hobson}. See~(\ref{opmvhs}) and~(\ref{radonvomm}) for details. For instance, in the classical Heston model, Hobson derives analytical formulas for $\tilde{\zeta}$. See~(\ref{hobeq}) in Section~\ref{capitulo:resultados}.

\

\noindent \textit{Hedging of fully path-dependent options}: The most interesting application of our results is the hedging of fully path-dependent options under stochastic volatility. For instance, if $H = \Phi(\{S_t; 0\le t\le T\})$ then Corollary~\ref{shiftTh} and Remark~\ref{fundremark} jointly with the above hedging methodologies allow us to dynamically hedge the payoff $H$ based on~(\ref{bhs}). The conditioning on the information flow $\{\mathcal{F}_{s_i}; s_i\in \Pi\}$ in the hedging strategy $\theta^{k,H}_{hedg}:=\{\theta^{k,H}_{s_i};s_i\in \Pi\}$ encodes the continuous monitoring of a path-dependent option. For each hedging date $s_i$, one has to incorporate the whole history of the price and volatility until such date in order to get an accurate description of the hedging. If $H$ is not path-dependent then the information encoded by $\{\mathcal{F}_{s_i}; s_i\in \Pi\}$ in $\theta^{k,H}_{hedg}$ is only crucial at time $s_i$.

\

Next, we provide the details of the Monte Carlo algorithm for the approximating pure hedging strategy $\theta^{k,H}_{hedg} = \{\theta^{k,H}_{s_i,0}; s_i\in \Pi\}$.

\section{The algorithm}\label{capitulo:algoritmo}
In this section we present the basic algorithm to evaluate the hedging strategy for a given European-type contingent claim $H \in L^2(\mathbb{Q})$ satisfying assumption \textbf{(M)} for a fixed $\mathbb{Q}\in \mathcal{M}^e$ at a terminal time $0 < T < \infty$.
The structure of the algorithm is based on the space-filtration discretization scheme induced by the stopping times $\{T^{k,j}_m; k \geq 1, m\ge 1, j=1, \ldots , p\}$. From the Markov property, the key point is the simulation of the first passage time $T^{k,j}_1$ for each $j=1\ldots, p$ for which we refer the work of Burq and Jones \cite{burq} for details.

\

\noindent \textbf{(Step 1)~Simulation of $\{A^{k,j} ; k\ge 1, j=1,\ldots,p\}$}.

\begin{enumerate}
\item One chooses $k\ge 1$ which represents the level of discretization of the Brownian motion.

\item One generates the increments $\{T^{k,j}_\ell-T^{k,j}_{\ell-1}; \ell \ge 1\}$ according to the algorithm described by Burq and Jones~\cite{burq}.

\item One simulates the family $\{\eta^{k,j}_\ell;\ell\ge 1 \}$ independently from $\{T^{k,j}_\ell-T^{k,j}_{\ell-1}; \ell \ge 1\}$. This $i.i.d$ family $\{\eta^{k,j}_\ell;\ell\ge 1\}$ must be simulated according to the Bernoulli random variable $\eta^{k,j}_1$ with parameter $1/2$ for $i=-1,1$. This simulates the jump process $A^{k,j}$ for $j=1,\ldots,p$.
\end{enumerate}

The next step is the simulation of $\mathbb{D}^{k,j}X_{T^{k,j}_1}$ where the conditional expectations in~(\ref{exder}) play a key role. For this, we need to simulate $H$ based on $\{S_t;0\le t\le T \}$ as follows.

\

\noindent \textbf{(Step 2)~Simulation of the risky asset price process $\{S^i;i=1,\ldots, d\}$}.

\begin{enumerate}
\item Generate a sample of $A^{k,i}$ according to Step 1 for a fixed $k\ge 1$.

\item With the partition $\mathcal{T}^k$ at hand, we can apply some appropriate approximation method to evaluate the discounted price. Generally speaking, we work with some It\^{o}-Taylor expansion method.
\end{enumerate}

The multidimensional setup requires an additional notation as follows. In the sequel, $t^{k,j}_\ell$ denotes the realization of the $T^{k,j}_\ell$ by means of Step 1, $t^{k}_{\ell}$ denotes the realization of $T^k_\ell$ based on the finest random partition $\mathcal{T}^k$. Moreover, any sequence $(t^k_1 < t^k_2 < \ldots < t^{k,j}_1)$ encodes the information generated by the realization of $\mathcal{T} ^k$ until the first hitting time of the $j$-th partition. In addition, we denote $t^{k,j}_{1-}$ as the last time in the finest partition previous to $t^{k,j}_1$. Let $\nu^{\ell}_k = (\nu^{\ell}_{1,k},\nu^{\ell}_{2,k})$ be the pair which realizes

$$t^k_\ell =t^{k,\nu^{\ell}_{1,k}}_{\nu^{\ell}_{2,k}},~k,\ell \ge 1$$
Based on this quantities, we define $\overline{\eta}^{k}_{t^k_\ell}$ as the realization of the random variable $\eta^{k,\nu^{\ell}_{1,k}}_{\nu^{\ell}_{2,k}}$. Recall expression \eqref{sigmakn}.

\

\noindent \textbf{(Step 3)~Simulation of the stochastic derivative}~$\mathbb{D}^{k,j}X_{T^{k,j}_1}$.

Based on Steps 1 and 2, for each $j=1,\ldots, p$ one simulates $\mathbb{D}^{k,j}X_{T^{k,j}_1}$ as follows. In the sequel, $\hat{\mathbb{E}}$ denotes the conditional expectation computed in terms of the Monte Carlo method:

\begin{equation}\label{montecarloder}
\hat{\mathbb{D}}^{k,j}_{t^{k,j}_1}X:= \frac{1}{2^{-k}\eta^{k,j}_1}\Big\{\hat{\mathbb{E}}\left[H\Big|\left(t^{k}_1,\overline{\eta}^k_{t^k_1}\right),\ldots, \left(t^{k,j}_1,\overline{\eta}^k_{t^{k,j}_1}\right)\right]
-\hat{\mathbb{E}}\left[H\Big|\left(t^{k}_1,\overline{\eta}^k_{t^k_1}\right),\ldots, \left(t^{k,j}_{1-},\overline{\eta}^k_{t^{k,j}_{1-}}\right)\right]\Big\},
\end{equation}
where with a slight abuse of notation, $\eta^{k,j}_1$ in \eqref{montecarloder} denotes the realization of the Bernoulli variable $\eta^{k,j}_1$. Then we define

\begin{equation}\label{eq:aproximacao_phi}
\hat{\phi}^{H,S,k}_0  := \left(\hat{\mathbb{D}}^{k,1}_{t^{k,1}_1} X,\ldots,\hat{\mathbb{D}}^{k,d}_{t^{k,d}_1} X\right),
\quad
\end{equation}


\

The correspondent simulated pure hedging strategy is given by

\begin{equation}\label{eq:estrategia_hedging_aproximada}
\bar{\theta}^{k,H}_{0,0}:=(\hat{\phi}^{H,S,k}_0)^\top \left[\diag(S_0)\sigma_0 \right]^{-1}.
\end{equation}

\

\noindent \textbf{(Step 4)~Simulation of}~$\theta^{k,H}_{hedg}$.

Repeat these steps several times and

\begin{equation}\label{finalH}
\hat{\theta}^{k,H}_{0,0}:= \text{mean of}~\bar{\theta}^{k,H}_{0,0}.
\end{equation}
The quantity~(\ref{finalH}) is a Monte Carlo estimate of $\theta^{k,H}_{0,0}$.

\begin{remark}
In order to compute the hedging strategy $\theta^{k,H}_{hedg}$ over a trading period $\{s_i;i=0, \ldots, q\}$, one perform the algorithm described above but based on the shifted filtration and the Brownian motions $W_{s_i}^{(j)}$ for $j=1,\ldots, p$ as described in Section~\ref{dished}.
\end{remark}

\begin{remark}
In practice, one has to calibrate the parameters of a given stochastic volatility model based on liquid instruments such as vanilla options and volatility surfaces. With those parameters at hand, the trader must follow the steps~(\ref{montecarloder}) and~(\ref{finalH}). The hedging strategy is then given by calibration and the computation of the quantity~(\ref{finalH}) over a trading period.
\end{remark}

\section{Numerical Analysis and Discussion of the Methods}\label{capitulo:resultados}
In this section, we provide a detailed analysis of the numerical scheme proposed in this work.

\subsection{Multidimensional Black-Scholes model}
At first, we consider the classical multidimensional Black-Scholes model with as many risky stocks as underlying independent random factors to be hedged $(d=p)$. In this case, there is only one equivalent local martingale measure, the hedging strategy $\theta^H$ is given by \eqref{GHedge_Strategy} and the cost is just the option price. To illustrate our method, we study a very special type of exotic option: a BLAC (Basket Lock Active Coupon) down and out barrier option whose payoff is given by

\[
H = \prod_{i\neq j}1\!\!1_{\{\min_{s\in[0,T]}S^i_s\vee\min_{s\in[0,T]}S^j_s>L\}}.
\]
It is well-known that for this type of option, there exists a closed formula for the hedging strategy. Moreover, it satisfies the assumptions of Theorem~\ref{result1}. See e.g~Bernis, Gobet and Kohatsu-Higa~\cite{kohatsu} for some formulas.

For comparison purposes with Bernis, Gobet and Kohatsu-Higa~\cite{kohatsu}, we consider $d =5$ underlying assets, $r = 0\%$ for the interest rate and $T=1$ year for the maturity time. For each asset, we set initial values $S^i_0 = 100;~1\le i\le 5$ and we compute the hedging strategy with respect to the first asset $S^1$ with discretization level $k = 3, 4, 5, 6$ and $20000$ simulations.

Following the work~\cite{kohatsu}, we consider the volatilities of the assets given by $\|\sigma^1\| = 35\%$, $\|\sigma^2\| = 35\%$, $\|\sigma^3\| = 38\%$, $\|\sigma^4\| = 35\%$ and $\|\sigma^5\| = 40\%$, the correlation matrix defined by $\rho_{ij}=0,4$ for $i\neq j$, where $\sigma^i = (\sigma_{i1},\cdots,\sigma_{i5})^\top$ and we use the barrier level $L = 76$. Table \ref{tabela:resultado_blac} provides the numerical results based on the algorithm described in Section \ref{capitulo:algoritmo} for the pointwise hedging strategy $\theta^H$. Due to Theorem~\ref{result1}, we expect that when the discretization level $k$ increases, we obtain results closer to the true value and this is what we find in our Monte Carlo experiments. The standard deviation and percentage $\%$ error in Table~\ref{tabela:resultado_blac} are related to the average of the hedging strategies calculated by Monte Carlo and the difference between the true and the estimated hedging value, respectively.

\begin{table}[h!]
    \centering
    \footnotesize
    \begin{tabular}{ccccccccccccc}
        \toprule[2pt]
         & & & & & & & & & & & &  \\
        & \textbf{k} &  & \textbf{Result} & & \textbf{St. error} & & \textbf{True value} & & \textbf{Diference} & & \textbf{$\%$ error} & \\
         & & & & & & & & & & & & \\ \hline
         & & & & & & & & & & & & \\
        & $3$ & & $0.00376$ & & $2.37\times 10^{-5}$ & & $0.00338$ & & $0.00038$ & & $11.15\%$ &\\
        & $4$ & & $0.00365$ & & $4.80\times 10^{-5}$ & & $0.00338$ & & $0.00027$ & & $8.03\%$ & \\
        & $5$ & & $0.00366$ & & $9.31\times 10^{-5}$ & & $0.00338$ & & $0.00028$ & & $8.35\%$ & \\
        & $6$ & & $0.00342$ & & $1.82\times 10^{-4}$ & & $0.00338$ & & $0.00004$ & & $1.29\%$ &\\
        & & & & & & & & & & & & \\
        \bottomrule[2pt]
    \end{tabular}
    \caption{\footnotesize Monte Carlo hedging strategy of a BLAC down and out option for a 5-dimensional Black-Scholes model.}
    \label{tabela:resultado_blac}
\end{table}

In Figure \ref{figura:resultado_blac3}, we plot the average hedging estimates with respect to the number of simulations. One should notice that when $k$ increases, the standard error also increases, which suggests more simulations for higher values of $k$.

\begin{figure}[h!]
\centering
\includegraphics[scale=0.4]{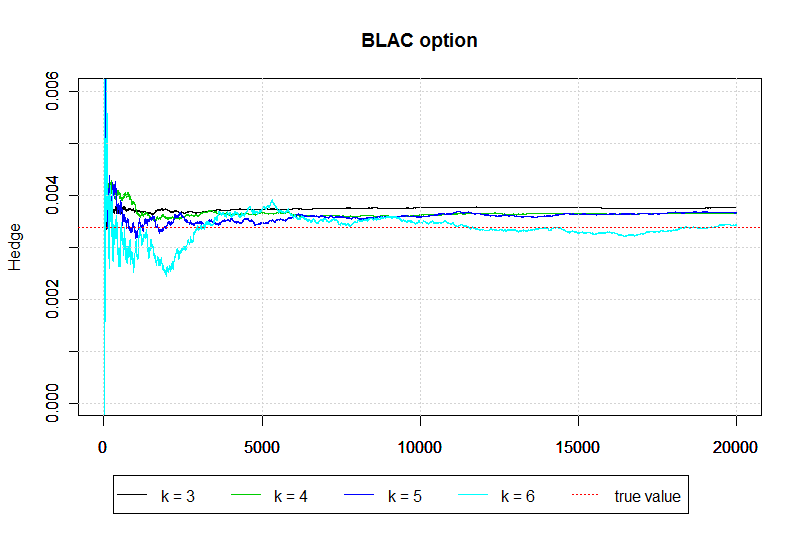}
\caption{\footnotesize Monte Carlo hedging strategy of a BLAC down and out option for a 5-dimensional Black-Scholes model.}
\label{figura:resultado_blac3}
\end{figure}

\subsection{Hedging Errors.}

Next, we present some hedging error results for two well-known non-constant volatility models: The constant elasticity of variance (CEV) model and the classical Heston stochastic volatility model~\cite{heston}. The typical examples we have in mind are the generalized F\"{o}llmer-Schweizer, local risk minimization and mean variance hedging strategies, where the optimal hedging strategies are computed by means of the minimal martingale measure and the variance optimal martingale measure, respectively. We analyze digital and one-touch one-dimensional European-type contingent claims as follows

$$\textbf{Digital option:}~~H = 1\!\!1_{\{S_T<95\}},\quad\textbf{One-touch option:}~~H = 1\!\!1_{\{\max_{t\in[0,T]}S_t>105\}}.$$

\

By using the algorithm described in Section~\ref{capitulo:algoritmo}, we compute the error committed by approximating the payoff $H$ by $\widehat{\mathbb{E}_{\mathbb{Q}}}[H]+\sum_{i=0}^{n-1}\hat{\theta}^{k,H}_{t_i,0} (S_{t_i,t_{i+1}-t_i}-S_{t_{i},0})$. This error will be called \emph{hedging error}. The computation of this error is summarized in the following steps:

\

\begin{enumerate}
\item We first simulate paths under the physical measure and compute the payoff $H$.


\item Then, we consider some deterministic partition of the interval [0,T] into $n$ points $t_0, t_1, \ldots, t_{n-1}$ such that $t_{i+1}-t_{i} = \frac{T}{n}$, for $i = 0, \ldots, n-1$.

\item One simulates at time $t_0 = 0$ the option price $\widehat{\mathbb{E}_{\mathbb{Q}}}[H]$ and the initial hedging estimate $\hat{\theta}^{k,H}_{0,0}$ from~\eqref{eq:aproximacao_phi}, \eqref{eq:estrategia_hedging_aproximada} and~(\ref{finalH}) under a fixed $\mathbb{Q}\in \mathcal{M}^e$ following the algorithm described in Section \ref{capitulo:algoritmo}.

\item We simulate $\hat{\theta}^{k,H}_{t_i,0}$ by means of the shifting argument based on the strong Markov property of the Brownian motion as described in Section~\ref{dished}.

\item We compute $\hat{H}$ by
\begin{equation}\label{eq:aproximacao_H}
\hat{H}:=\widehat{\mathbb{E}_{\mathbb{Q}}}[H]+\sum_{i=0}^{n-1}\hat{\theta}^{k,H}_{t_i,0}(S_{t_i,t_{i+1}-t_i}-S_{t_i,0}).
\end{equation}

\item Finally, the hedging error estimate $\gamma$ and the percentual error $e_{\gamma}$ are given by $\gamma := H-\hat{H}$ and $e_{\gamma} := 100\times \gamma/\widehat{\mathbb{E}_{\mathbb{Q}}}[H]$, respectively.
\end{enumerate}

\begin{remark}\label{expl}
When no locally-risk minimizing strategy is available, we also expect to obtain low hedging errors when dealing with generalized F\"{o}llmer-Schweizer decompositions due to the orthogonal martingale decomposition. In the mean variance hedging case, two terms appear in the optimal hedging strategy: the pure hedging component $\theta^{H,\tilde{\mathbb{P}}}$ of the GKW decomposition under the optimal variance martingale measure $\tilde{\mathbb{P}}$ and $\tilde{\zeta}$ as described by~(\ref{opmvhs}) and~(\ref{radonvomm}). For the Heston model, $\tilde{\zeta}$ was explicitly calculated by Hobson~\cite{hobson}. We have used his formula in our numerical simulations jointly with $\hat{\theta}^{k,H}$ under $\tilde{\mathbb{P}}$ in the calculation of the mean variance hedging errors. See expression~(\ref{hobeq}) for details.
\end{remark}

\subsubsection{Constant Elasticity of Variance (CEV) model}
The discounted risky asset price process described by the CEV model under the physical measure is given by
\begin{equation}\label{eq:modelo_cev}
dS_t = S_t\left[(b_t-r_t)dt+\sigma S^{(\beta-2)/2}_tdB_t\right], \quad S_0 = s,
\end{equation}
where $B$ is a $\mathbb{P}$-Brownian motion. The instantaneous sharpe ratio is $\psi_t = \frac{b_t-r_t}{\sigma S^{(\beta-2)/2}_t}$ such that the model can be rewritten as
\begin{equation}\label{eq:modelo_cev_mme}
dS_t = \sigma_t S^{\beta/2}_tdW_t
\end{equation}
where $W$ is a $\mathbb{Q}$-Brownian motion and $\mathbb{Q}$ is the equivalent local martingale measure. For both the digital and one-touch options, we consider the parameters $r = 0$ for the interest rate, $T = 1$ (month) for the maturity time, $\sigma = 0.2$, $S_0 = 100$ and $\beta = 1.6$ such that the constant of elasticity is $-0.4$. We simulate the hedging error along $[0,T]$ considering discretization levels $k=3, 4$, and $1$ and $2$ hedging strategies per day, which means approximately $22$ and $44$ hedging strategies, respectively, along the interval $[0,T]$. From Corollary~\ref{shiftTh}, we know that this procedure is consistent. For the digital option, we also recall that the hedging strategy has continuous paths up to some stopping time~(see Zhang~\cite{zhang}) so that Theorem~\ref{result1} and Remark~\ref{smoothre} apply accordingly. The hedging error results for the digital and one-touch options are summarized in Tables~\ref{tabela:hedging_error_digital_cev} and~\ref{tabela:hedging_error_barreira_cev}, respectively. The standard deviations are related to the hedging errors.

\begin{table}[h!]
    \scriptsize
    \centering
    \begin{tabular}{ccccccccccccccc}
        \toprule[2pt]
         & & & & & & & & & & & & & &  \\
        & \textbf{Simulations} &  & \textbf{k} & & \textbf{Hedges/day} & & \textbf{Hedging error}~$\gamma$ & & \textbf{St. dev.} & & \textbf{Price } & & \textbf{\% Error}~$e_\gamma$ & \\
         & & & & & & & & & & & & & & \\ \hline
         & & & & & & & & & & & & & & \\
        & $200$ & & $3$ & & $1$ & & $-0.02696$ & & $0.1750$ & & $0.2864$ & & $9.41\%$ & \\
        & $200$ & & $3$ & & $2$ & & $-0.00473$ & & $0.1451$ & & $0.2864$ & & $1.65\%$ & \\
        & $200$ & & $4$ & & $1$ & & $ 0.00494$ & & $0.1562$ & & $0.2759$ & & $1.79\%$ & \\
        & $200$ & & $4$ & & $2$ & & $-0.00291$ & & $0.1522$ & & $0.2760$ & & $1.05\%$ & \\
        & & & & & & & & & & & & & &\\
        \bottomrule[2pt]
    \end{tabular}
    \caption{Hedging error of a digital option for the CEV model.}
    \label{tabela:hedging_error_digital_cev}
\end{table}

\begin{table}[h!]
    \scriptsize
    \centering
    \begin{tabular}{ccccccccccccccc}
        \toprule[2pt]
         & & & & & & & & & & & & & &  \\
        & \textbf{Simulations} &  & \textbf{k} & & \textbf{Hedges/day} & & \textbf{Hedging error}~$\gamma$ & & \textbf{St. dev.} & & \textbf{Price} & & \textbf{\% Error}~$e_\gamma$ & \\
         & & & & & & & & & & & & & & \\ \hline
         & & & & & & & & & & & & & & \\
        & $600$ & & $3$ & & $1$ & & $0.0417$ & & $0.1727$ & & $0.4804$ & & $8.68\%$ & \\
        & $600$ & & $3$ & & $2$ & & $0.0424$ & & $0.1413$ & & $0.4804$ & & $8.82\%$ & \\
        & $600$ & & $4$ & & $1$ & & $0.0144$ & & $0.1770$ & & $0.5061$ & & $2.84\%$ & \\
        & $600$ & & $4$ & & $2$ & & $0.0125$ & & $0.1168$ & & $0.5060$ & & $2.47\%$ & \\
        & & & & & & & & & & & & & &\\
        \bottomrule[2pt]
    \end{tabular}
    \caption{Hedging error of one-touch option for the CEV model.}
    \label{tabela:hedging_error_barreira_cev}
\end{table}

\subsubsection{Heston's Stochastic Volatility Model}

Here we consider two types of hedging methodologies: Local-risk minimization and mean variance hedging strategies as described in the Introduction and Remark~{\ref{expl}}. The Heston dynamics of the discounted price under the physical measure is given by

\[
\left\{\begin{array}{l}
dS_t = S_t(b_t-r_t)\Sigma_tdt + S_t\sqrt{\Sigma_t}dB^{(1)}_t \\
d\Sigma_t = 2\kappa (\theta-\Sigma_t)dt + 2\sigma\sqrt{\Sigma_t}dZ_t,0\le t\le T,
\end{array}\right.
\]
where $Z = \rho B^{(1)} + \bar{\rho}B^{(2)}_t$, $\bar{\rho} = \sqrt{1-\rho^2}$, with $(B^{(1)}B^{(2)})$ two independent $\mathbb{P}$-Brownian motions  and $\kappa, m, \beta_0, \mu$ are suitable constants in order to have  a well-defined Markov process~(see e.g~Heston~\cite{heston}). Alternatively, we can rewrite the dynamics as

\[
\left\{\begin{array}{l}
dS_t = S_t Y^2_t (b_t-r_t) dt + S_tY_tdB^{(1)}_t\\
dY_t = \kappa \Bigg (\frac{m}{Y_t} - Y_t\Bigg)dt + \sigma dZ_t,~0\le t\le T,
\end{array}\right.
\]
where $Y=\sqrt{\Sigma_t}$ and $m = \theta - \frac{\sigma^2}{2\kappa}$.

\

\noindent \textbf{Local-Risk Minimization}.
For comparison purposes with~Heath, Platen and Schweizer~\cite{Heath}, we consider the hedging of a European put option $H$ written on a Heston model with correlation parameter $\rho=0$. We set $S_0 = 100$, strike price $K = 100$, $T = 1$ (month) and we use discretization levels $k = 3, 4$ and $5$. We set the parameters $\kappa = 2.5$, $\theta = 0.04$, $\rho = 0$, $\sigma = 0.3$, $r = 0$ and $Y_0 = 0.02$. In this case, the hedging strategy $\theta^{H,\hat{\mathbb{P}}}$ based on the local-risk-minimization methodology is bounded with continuous paths so that Theorem~\ref{result1} applies to this case. Moreover, as described by Heath, Platen and Schweizer~\cite{Heath}, $\theta^{H,\hat{\mathbb{P}}}$ can be obtained by a PDE numerical analysis.

Table \ref{tabela:resultado_heston_schweizer} presents the results of the hedging strategy $\hat{\theta}^{k,H}_{0,0}$ by using the algorithm described in Section~\ref{capitulo:algoritmo}. Figure~\ref{resultado_heston1_hedge3} provides the Monte Carlo hedging strategy with respect to the number of simulations of order $10000$. We notice that our results agree with the results obtained by Heath, Platen and Schweizer~\cite{Heath} by PDE methods. In this case, the true value of the hedging at time $t = 0$ is approximately $-0.44$. The standard errors in Table~\ref{tabela:resultado_heston_schweizer} are related to the hedging and prices computed, respectively, from the Monte Carlo method described in Section~\ref{capitulo:algoritmo}.

\begin{table}[h!]
    \centering
    \footnotesize
    \begin{tabular}{ccccccccccc}
        \toprule[2pt]
         & & & & & & & & & & \\
         & \textbf{k} &  & \textbf{Hedging} & & \textbf{Standard error} & & \textbf{Monte Carlo price} & & \textbf{Standard error} & \\
         & & & & & & & & & & \\ \hline
         & & & & & & & & & & \\
         & $3$ & & $-0.4480$ & & $6.57\times 10^{-4}$ & & $10.417$ & & $5.00 \times 10^{-3}$ & \\
         & $4$ & & $-0.4506$ & & $1.28\times 10^{-3}$ & & $10.422$ & & $3.35 \times 10^{-3}$ & \\
         & $5$ & & $-0.4453$ & & $2.54\times 10^{-3}$ & & $10.409$ & & $2.75 \times 10^{-3}$ &\\
         & & & & & & & & & & \\
        \bottomrule[2pt]
    \end{tabular}
    \caption{\footnotesize Monte Carlo local-risk minimization hedging strategy of a European put option with Heston model.}
    \label{tabela:resultado_heston_schweizer}
\end{table}

\begin{figure}[h!]
\centering
\includegraphics[scale=0.5]{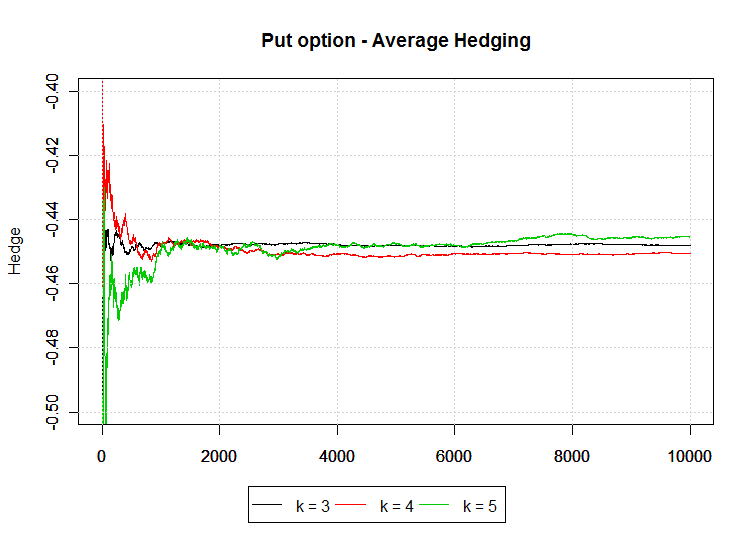}
\caption{\footnotesize Monte Carlo local-risk minimization hedging strategy of a European put option with Heston model.}
\label{resultado_heston1_hedge3}
\end{figure}

\

\noindent \textbf{Hedging with generalized F\"{o}llmer-Schweizer decomposition for one-touch option.} Based on Corollary~\ref{shiftTh}, we also present the hedging error associated to one-touch options for a Heston model with non-zero correlation. We simulate the hedging error along the interval $[0,1]$ using $k = 3, 4$ as discretization levels and $1$ and $2$ hedging strategies per day with parameters $\kappa = 3.63$, $\theta = 0.04$, $\rho = -0.53$, $\sigma = 0.3$, $r = 0$, $b=0.01$, $Y_0 = 0.3$ and $S_0=100$ where the barrier is $105$. The hedging error result for the one-touch option is summarized in Table~\ref{tabela:hedging_error_barreira_heston}. The standard deviations in Table~\ref{tabela:hedging_error_barreira_heston} are related to the hedging error.


To our best knowledge, there is no result about the existence of locally-risk minimizing hedging strategies for one-touch options written on a Heston model with nonzero correlation. As pointed out in Remark~\ref{expl}, it is expected that pure hedging strategies based on the generalized F\"{o}llmer-Schweizer decomposition mitigate very-well the hedging error. This is what we get in the simulation results.


\begin{table}[h!]
    \scriptsize
    \centering
    \begin{tabular}{ccccccccccccccc}
        \toprule[2pt]
         & & & & & & & & & & & & & &  \\
        & \textbf{Simulations} &  & \textbf{k} & & \textbf{Hedges/day} & & \textbf{Hedging error}~$\gamma$ & & \textbf{St. dev.} & & \textbf{Price} & & \textbf{\% error}~$e_\gamma$ & \\
         & & & & & & & & & & & & & & \\ \hline
         & & & & & & & & & & & & & & \\
        & $600$ & & $3$ & & $1$ & & $0.0409$ & & $0.2452$ & & $0.7399$ & & $5.53\%$ & \\
        & $600$ & & $3$ & & $2$ & & $0.0316$ & & $0.2450$ & & $0.7397$ & & $4.27\%$ & \\
        & $600$ & & $4$ & & $1$ & & $0.0268$ & & $0.2842$ & & $0.7735$ & & $3.46\%$ & \\
        & $600$ & & $4$ & & $2$ & & $0.0191$ & & $0.2605$ & & $0.7738$ & & $2.47\%$ & \\
        & & & & & & & & & & & & & &\\
        \bottomrule[2pt]
    \end{tabular}
    \caption{Hedging error with generalized F\"{o}llmer-Schweizer decomposition: One-touch option with Heston model.}
    \label{tabela:hedging_error_barreira_heston}
\end{table}

\

\noindent \textbf{Mean variance hedging strategy}. Here we present the hedging errors associated to one-touch options written on a Heston model with non-zero correlation under the mean variance methodology. Again, we simulate the hedging error along the interval $[0,1]$ using $k = 3, 4$ as discritization levels and $1$ and $2$ hedging strategies per day with parameters $r=0$, $b=0.01$, $\kappa = 3.63$, $\theta = 0.04$, $\rho = -0.53$, $\sigma = 0.3$, $Y_0 = 0.3$ and $S_0=100$ with barrier 105. The computation of the optimal hedging strategy follows from Remark~\ref{expl}. The quantity $\tilde{\zeta}$ is not related to the GKW decomposition but it is described by Theorem 1.1 in Hobson \cite{hobson} as follows. The process $\tilde{\zeta}$ appearing in~(\ref{opmvhs}) and~(\ref{radonvomm}) is given by

\begin{equation}\label{hobeq}
\tilde{\zeta}_t = \tilde{Z}_0\rho\sigma F(T-t)-\tilde{Z}_0b;~0\le t\le T,
\end{equation}
where $F$ is given by (see case 2 of Prop. 5.1 in Hobson \cite{hobson})
$$
F(t) = \frac{C}{A}\tanh\left(ACt+\tanh^{-1}\left(\frac{AB}{C}\right)\right)-B;~0\le t\le T,
$$
with $A = \sqrt{|1-2\rho^2|\sigma^2}$, $B = \frac{\kappa+2\rho\sigma b}{\sigma^2|1-2\rho^2|}$ and $C = \sqrt{|D|}$ where $D = 2b^2+\frac{(\kappa+2\rho\sigma b)^2)}{\sigma^2(1-2\rho^2)}$. The initial condition $\tilde{Z}_0$ is given by

$$
\tilde{Z}_0=\frac{Y_0^2}{2}F(T) + \kappa\theta\int_0^TF(s)ds.
$$

The hedging error results are summarized in Table \ref{tabela:hedging_error_barreira_heston_mean_variance} where the standard deviations are related to the hedging error. In comparison with the local-risk minimization methodology, the results show smaller percentual errors when $k$ increases. Also, in all the cases, we had smaller values of the standard deviation which suggests the mean variance methodology provides more accurate values of the hedging strategy.

\begin{table}[h!]
    \scriptsize
    \centering
    \begin{tabular}{ccccccccccccccc}
        \toprule[2pt]
         & & & & & & & & & & & & & &  \\
        & \textbf{Simulations} &  & \textbf{k} & & \textbf{Hedges/day} & & \textbf{Hedging error}~$\gamma$ & & \textbf{St. dev.} & & \textbf{Price} & & \textbf{\% error}~$e_\gamma$ & \\
         & & & & & & & & & & & & & & \\ \hline
         & & & & & & & & & & & & & & \\
        & $600$ & & $3$ & & $1$ & & $0.0689$ & & $0.1688$ & & $0.7339$ & & $9.39\%$ & \\
        & $600$ & & $3$ & & $2$ & & $0.0592$ & & $0.1344$ & & $0.7339$ & & $8.07\%$ & \\
        & $600$ & & $4$ & & $1$ & & $0.0213$ & & $0.1846$ & & $0.7766$ & & $2.74\%$ & \\
        & $600$ & & $4$ & & $2$ & & $0.0161$ & & $0.1278$ & & $0.7765$ & & $2.07\%$ & \\
        & & & & & & & & & & & & & &\\
        \bottomrule[2pt]
    \end{tabular}
    \caption{Hedging error in the mean variance hedging methodology for one-touch option with Heston model.}
    \label{tabela:hedging_error_barreira_heston_mean_variance}
\end{table}

\section{Appendix}
This appendix provides a deeper understanding of the Monte Carlo algorithm proposed in this work when the representation $(\phi^{H,S}, \phi^{H,I})$ in~(\ref{GHedge_Strategy}) admits additional integrability and path smoothness assumptions. We present stronger approximations which complement the asymptotic result given in Theorem~\ref{deltaj}. Uniform-type weak and strong pointwise approximations for $\theta^H$ are presented and they validate the numerical experiments in Tables 1 and 4 in Section~\ref{capitulo:resultados}. At first, we need of some technical lemmas.

\begin{lemma}\label{identitytkn}
Suppose that $\phi^H= (\phi^{H,1}, \ldots ,\phi^{H,p})$ is a $p$-dimensional progressive process such that $\mathbb{E}\sup_{0\le t\le T}\|\phi^{H}_t\|^2_{\mathbb{R}^p} < \infty$. Then, the following identity holds

\begin{equation} \label{rfdme}
   \Delta \delta^{k} X_{ T^{k,j}_{1} }  = \mathbb{E} \left[ \int_0^{T^{k,j}_{1}} \phi^{H,j}_s dW^{(j)}_s \mid \mathcal{F}^k_{T^{k,j}_{1}} \right]~a.s;~j=1, \ldots , p; k\ge 1.
\end{equation}
\end{lemma}

\begin{proof}
It is sufficient to prove for $p=2$ since the argument for $p > 2$ easily follows from this case. Let $\mathcal{H}$ be the linear space constituted by the bounded $\mathbb{R}^2$-valued $\mathbb{F}$-progressive processes $\phi = (\phi^1,\phi^2)$ such that \eqref{rfdme} holds with $X = X_0 + \int_0^\cdot \phi^1_sdW^{(1)}_s + \int_0^\cdot \phi^2_sdW^{(2)}_s$ where $X_0\in \mathcal{F}_0$. Let $\mathcal{U}$ be the class of stochastic intervals of the form $[[S, +\infty[[$ where $S$ is a $\mathbb{F}$-stopping time. We claim that $\phi=\big(1\!\!1_{[[S, + \infty[[}, 1\!\!1_{[[J, + \infty[[}\big) \in \mathcal{H}$ for every $\mathbb{F}$-stopping times $S$ and $J$. In order to check~(\ref{rfdme}) for such $\phi$, we only need to show for $j=1$ since the argument for $j=2$ is the same. With a slight abuse of notation, any sub-sigma algebra of $\mathcal{F}_T$ of the form $\Omega^*_1\otimes \mathcal{G}$ will be denoted by $\mathcal{G}$ where $\Omega^*_1$ is the trivial sigma-algebra on the first copy $\Omega_1$.

At first, we split $\Omega = \bigcup_{n=1}^\infty \{T^{k}_n = T^{k,1}_1 \}$ and we make the argument on the sets $\{T^k_n=T^{k,1}_1 \};~n\ge 1$. In this case, we know that $\mathcal{F}^k_{T^{k,1}_1} = \mathcal{F}^{k,1}_{T^{k,1}_1}\otimes \mathcal{F}^{k,2}_{T^{k,2}_{n-1}} \ $~a.s and

\[
\Delta \delta^{k} X_{ T^{k,j}_{1} } =
\Delta \delta^{k}  \left( W^{(1)}_{T^{k,1}_1} - W^{(1)}_S \right) 1\!\!1_{ \{ S < T^{k,1}_1  \} } + \Delta \delta^{k} \left( W^{(2)}_{T^{k,1}_1} - W^{(2)}_J \right) 1\!\!1_{ \{ J < T^{k,1}_1  \} }.
\]
The independence between $W^{(1)}$ and $W^{(2)}$ and the independence of the Brownian motion increments yield

\begin{align*}
\Delta \delta^{k} \left( W^{(2)}_{T^{k,1}_1} - W^{(2)}_J \right) & = \mathbb{E} \left( W^{(2)}_{T^{k,1}_1} - W^{(2)}_J \mid \mathcal{F}^k_{T^{k,1}_1} \right) - \mathbb{E} \left( W^{(2)}_{T^{k,1}_1} - W^{(2)}_J \mid \mathcal{F}^k_{T^{k}_{n-1}} \right) \\
& = \mathbb{E} \left( W^{(2)}_{T^{k,1}_1} - W^{(2)}_J \mid \mathcal{F}^{k,1}_{T^{k,1}_1}\otimes \mathcal{F}^{k,2}_{T^{k,2}_{n-1}} \right)  = \mathbb{E} \left( W^{(2)}_{T^{k,1}_1} - W^{(2)}_J \mid \sigma \{T^{k,1}_1\} \otimes \mathcal{F}^{k,2}_{T^{k,2}_{n-1}} \right) \\
& = \mathbb{E} \left( W^{(2)}_{T^{k,1}_1} - W^{(2)}_J \mid  \mathcal{F}^{k,2}_{T^{k,2}_{n-1}} \right)  = 0 \quad a.s
\end{align*}
on the set $\{T_{n-1}^k \leq J < T_{n}^k=T^{k,1}_1 \}$. We also have,

\begin{align*}
\Delta \delta^{k} \left( W^{(2)}_{T^{k,1}_1} - W^{(2)}_J \right) & = \mathbb{E} \left( W^{(2)}_{T^{k,1}_1} - W^{(2)}_J \mid \mathcal{F}^k_{T^{k,1}_1} \right) - \mathbb{E} \left( W^{(2)}_{T^{k,1}_1} - W^{(2)}_J \mid \mathcal{F}^k_{T^{k}_{n-1}} \right) \\
& = \mathbb{E} \left( W^{(2)}_{T^{k,1}_1} - W^{(2)}_J \mid \mathcal{F}^{k,1}_{T^{k,1}_1}\otimes \mathcal{F}^{k,2}_{T^{k,2}_{n-1}} \right) - \mathbb{E} \left( W^{(2)}_{T^{k}_{n-1}} - W^{(2)}_J \mid \mathcal{F}^{k,2}_{T^{k,2}_{n-1}} \right) \\
& = \mathbb{E} \left( W^{(2)}_{T^{k,1}_1} - W^{(2)}_J \mid \sigma \{T^{k,1}_1\} \otimes \mathcal{F}^{k,2}_{T^{k,2}_{n-1}} \right) - \mathbb{E} \left( W^{(2)}_{T^{k}_{n-1}} - W^{(2)}_J \mid \mathcal{F}^{k,2}_{T^{k,2}_{n-1}} \right) \\
& = \mathbb{E} \left( W^{(2)}_{T^{k}_{n-1}} - W^{(2)}_J \mid  \mathcal{F}^{k,2}_{T^{k,2}_{n-1}} \right) - \mathbb{E} \left( W^{(2)}_{T^{k}_{n-1}} - W^{(2)}_J \mid \mathcal{F}^{k,2}_{T^{k,2}_{n-1}} \right) = 0,
\end{align*}
on the set $\{J < T_{n-1}^k \}$. By construction $\mathcal{F}^k_{T^{k,1}_1} = \mathcal{F}^{k,1}_{T^{k,1}_1}\otimes \mathcal{F}^{k,2}_{T^{k,2}_{n-1}} \ $~a.s and again the independence between
$W^{(1)}$ and $W^{(2)}$ yields

\begin{align*}
\Delta \delta^{k} \left( W^{(1)}_{T^{k,1}_1} - W^{(1)}_S \right) & = \mathbb{E} \left( W^{(1)}_{T^{k,1}_1} - W^{(1)}_S \mid \mathcal{F}^k_{T^{k,1}_1} \right) - \mathbb{E} \left( W^{(1)}_{T^{k,1}_1} - W^{(1)}_S \mid \mathcal{F}^k_{T^{k}_{n-1}} \right) \\
& = \mathbb{E} \left( W^{(1)}_{T^{k,1}_1} - W^{(1)}_S \mid \mathcal{F}^k_{T^{k,1}_1} \right)
\end{align*}
on $\{T_{n-1}^k \leq S < T_{n}^k=T^{k,1}_1 \}$. Similarly,


\begin{align*}
\Delta \delta^{k} \left( W^{(1)}_{T^{k,1}_1} - W^{(1)}_S \right) & = \mathbb{E} \left( W^{(1)}_{T^{k,1}_1} - W^{(1)}_S \mid \mathcal{F}^k_{T^{k,1}_1} \right) - \mathbb{E} \left( W^{(1)}_{T^{k,1}_1} - W^{(1)}_S \mid \mathcal{F}^k_{T^{k}_{n-1}} \right) \\
& = \mathbb{E} \left( W^{(1)}_{T^{k,1}_1} - W^{(1)}_S \mid \mathcal{F}^k_{T^{k,1}_1} \right) - \mathbb{E} \left( W^{(1)}_{T^{k}_{n-1}} - W^{(1)}_S \mid \mathcal{F}^{k,2}_{T^{k}_{n-1}} \right) \\
& = \mathbb{E} \left( W^{(1)}_{T^{k,1}_1} - W^{(1)}_S \mid \mathcal{F}^k_{T^{k,1}_1} \right) - \mathbb{E} \left( W^{(1)}_{T^{k}_{n-1}} \mid \mathcal{F}^{k,2}_{T^{k}_{n-1}}  \right) + \mathbb{E} \left( W^{(1)}_S \mid \mathcal{F}^{k,2}_{T^{k}_{n-1}} \right) \\
& = \mathbb{E} \left( W^{(1)}_{T^{k,1}_1} - W^{(1)}_S \mid \mathcal{F}^k_{T^{k,1}_1} \right)  + \mathbb{E} \left( W^{(1)}_S \mid \mathcal{F}^{k,2}_{T^{k}_{n-1}} \right)
\end{align*}
on $\{S < T_{n-1}^k\}$. By assumption $S$ is an $\mathbb{F}$-stopping time, where $\mathbb{F}$ is a product filtration. Hence,  $\mathbb{E} \big(W^{(1)}_S | \mathcal{F}^{k,2}_{T^{k}_{n-1}}\big)=0$ a.s on $\{S < T_{n-1}^k\}$.

Summing up the above identities, we shall conclude $\big(1\!\!1_{[[S, + \infty[[}, 1\!\!1_{[[J, + \infty[[}\big)\in \mathcal{H}$. In particular, the constant process $(1,1)\in \mathcal{H}$ and if $\phi^n$ is a sequence in $\mathcal{H}$ such that $\phi^n\rightarrow \phi$ a.s $Leb\times \mathbb{Q}$ with $\phi$ bounded, then a routine application of Burkh\"{o}lder inequality shows that $\phi \in \mathcal{H}$. Since $\mathcal{U}$ generates the optional sigma-algebra then we shall apply the monotone class theorem and, by localization, we may conclude the proof.




\end{proof}

\begin{lemma}\label{lh}
Let $B$ be a one-dimensional Brownian motion and $S^k_n:= \inf\{ t> S^k_{n-1}; |B_t-B_{S^k_{n-1}}|=2^{-k}\}$ with $S^k_0=0~a.s$,~$n\ge 1$. If $\varphi$ is an absolutely continuous and non-negative adapted process then there exists a deterministic constant $C$ which does not depend on $m,k\ge 1$ such that

$$\Bigg|\int_{S^k_{m-1}}^{S^k_m}\varphi_tdB_t\Bigg|^21\!\!1_{\{S^{k}_{m}\le T\}} \le C\sup_{0\le t\le T}|\varphi_t|^22^{-2k}~a.s;~k,m\ge 1.$$
\end{lemma}
\begin{proof}
For given $m,k\ge 1$, Young inequality and integration by parts yield
\begin{align*}
\Bigg|\int_{S^k_{m-1}}^{S^k_m}\varphi_tdB_t\Bigg|^2 & \le C \Big\{ |\varphi_{S^k_m}|^2|B_{S^k_m}|^2 +|\varphi_{S^k_{m-1}}|^2|B_{S^k_{m-1}}|^2 + \big|\int_{S^k_{m-1}}^{S^k_m}B_td\varphi_t \big|^2\Big\}\\
& \le C 2^{-2k}\sup_{0\le t\le T}|\varphi_t|^2 + C\sup_{S^k_{m-1}\le t\le S^k_m}|B_t|^2|Var(\varphi)_{S^k_{m}} - Var(\varphi)_{S^k_{m-1}}|^2\\
& \le C 2^{-2k}\sup_{0\le t\le T}|\varphi_t|^2 + C2^{-2k}|Var(\varphi)_{S^k_{m}} - Var(\varphi)_{S^k_{m-1}}|^2\\
& = C 2^{-2k}\sup_{0\le t\le T}|\varphi_t|^2 + C2^{-2k}|\varphi_{S^k_{m}} - \varphi_{S^k_{m-1}}|^2\\
& \le C2^{-2k}\sup_{0\le t\le T}|\varphi_t|^2~a.s~\text{on}~\{S^{k}_{m}\le T\},
\end{align*}
for some constant $C$ which does not depend on $m,k\ge 1$.
\end{proof}

\begin{lemma}\label{boundness1}
Assume that $\phi^{H,j}\in \text{B}^2(\mathbb{F})$ for some $j=1,\ldots,p$. Then there exists a constant $C$ such that

$$\sup_{k\ge 1}\mathbb{E}\sup_{0\le t\le T}|\mathbf{D}^{k,j}X_t|^2 \le C \mathbb{E}\sup_{0\le t\le T}|\phi^{H,j}|^2.$$

\end{lemma}
\begin{proof}
By repeating the argument employed in Lemma~\ref{identitytkn} for $k\ge 1$,~$n > 1$ and $j\in \{1,\ldots,p\}$, we shall write

$$\mathbf{D}^{k,j}X_t = \mathbb{E}\Bigg[ \frac{1}{\Delta A^{k,j}_{T^{k,j}_n}}\int_{T^{k,j}_{n-1}}^{T^{k,j}_{n}}\phi^{H,j}_tdW^{(j)}_t \big| \mathcal{F}^k_{T^{k,j}_{n-1}} \Bigg]~a.s~\text{on}~\{ T^{k,j}_{n-1} < t \le T^{k,j}_{n}\}.
$$
Doob maximal inequalities combined with Jensen inequality yield

\begin{equation}\label{doobd}
\mathbb{E}\sup_{0\le t\le T}|\mathbf{D}^{k,j}X_t|^2 \le C 2^{2k} \mathbb{E}\sup_{n\ge 1} \Bigg| \int_{T^{k,j}_{n-1}}^{T^{k,j}_n} \phi^{H,j}_tdW^{(j)}_t\Bigg|^21\!\!1_{\{T^{k,j}_{n}\le T\}},
\end{equation}
for $k\ge 1$ and for some positive constant $C$. Now, we need a path-wise argument in order to estimate the right-hand side of \eqref{doobd}. For this, let us define

$$\varphi^{\ell,j}_t:= \ell \int_{t-\frac{1}{\ell}}^{t}\phi^{H,j}_sds;~\ell\ge 1;~0\le t\le T.$$
Lemma~\ref{lh} and the fact that $\sup_{0\le t\le T}|\varphi^{\ell,j}_t|^2 \le \sup_{0\le t\le T}|\phi^{H,j}_t|^2~\forall \ell \ge 1$ yield

\begin{equation}\label{lh1}
\Bigg|\int_{T^{k,j}_{n-1}}^{T^{k,j}_n}\varphi^{\ell,j}_tdW^{(j)}_t\Bigg|^2 1\!\!1_{\{T^{k,j}_{n}\le T\}}\le C\sup_{0\le t\le T}|\phi^{H,j}_t|^22^{-2k};~\ell,n,k\ge 1,
\end{equation}
where $C$ is the constant in Lemma~\ref{lh}. Now, by applying Lemma 2.4 in Nutz~\cite{nutz}, the estimate \eqref{lh1} and a routine localization procedure, the following estimate holds

$$
\Bigg|\int_{T^{k,j}_{n-1}}^{T^{k,j}_n}\phi^{H,j}_tdW^{(j)}_t\Bigg|^21\!\!1_{\{T^{k,j}_{n}\le T\}} \le C\sup_{0\le t\le T}|\phi^{H,j}_t|^22^{-2k};~k\ge 1,
$$
and therefore

\begin{equation}\label{lh2}
\mathbb{E}\sup_{n\ge 1}\Bigg|\int_{T^{k,j}_{n-1}}^{T^{k,j}_n}\phi^{H,j}_tdW^{(j)}_t\Bigg|^2 1\!\!1_{\{T^{k,j}_{n}\le T\}}\le C\mathbb{E}\sup_{0\le t\le T}|\phi^{H,j}_t|^22^{-2k}~\forall k\ge 1.
\end{equation}
The estimate \eqref{doobd} combined with \eqref{lh2} allow us to conclude the proof if $\phi^{H,j} \ge 0~a.s~(Leb\times \mathbb{Q})$. By splitting $\phi^{H,j}= \phi^{H,j,+} - \phi^{H,j,-}$ into the negative and positive parts, we may conclude the proof of the lemma.
\end{proof}

The following result allows us to get a uniform-type weak convergence of $\mathbf{D}^{k,j}X$ under very mild integrability assumption.

\begin{theorem}\label{mainres}
Let $H$ be a $\mathbb{Q}$-square integrable contingent claim satisfying assumption~\textbf{(M)} and assume that $H$ admits a representation $\phi^{H}$ such that $\phi^{H,j} \in \text{B}^2(\mathbb{F})$ for some $j\in \{1,\ldots, p\}$. Then

$$\lim_{k\rightarrow \infty}\mathbf{D}^{k,j}X = \phi^{H,j}$$
weakly in $\text{B}^2(\mathbb{F})$.
\end{theorem}
\begin{proof}
Let us fix $j=1,\ldots,p$. From Lemma~\ref{boundness1}, we know that $\{ \mathbf{D}^{k,j}X; k\ge 1 \}$ is bounded in $\text{B}^2(\mathbb{F})$ and therefore this set is weakly relatively compact in $\text{B}^2(\mathbb{F})$. By Eberlein Theorem, we also know that it is $\text{B}^2(\mathbb{F})$-weakly relatively sequentially compact. From Theorem~\ref{deltaj},

$$\lim_{k\rightarrow \infty}\mathbf{D}^{k,j}X = \phi^{H,j}$$
weakly in $L^2(Leb\times \mathbb{Q})$ and since $\|\cdot\|_{\text{B}^2}$ is stronger than $\|\cdot\|_{L^2(Leb\times\mathbb{Q})}$, we necessarily have the full convergence
$$\lim_{k\rightarrow \infty}\mathbf{D}^{k,j}X = \phi^{H,j}$$
in $\sigma(\text{B}^2, \text{M}^2)$.
\end{proof}

Next, we analyze the pointwise strong convergence for our approximation scheme.

\subsection{Strong Convergence under Mild Regularity}\label{capitulo:aproximacao_forte}
In this section, we provide a pointwise strong convergence result for GKW projectors under rather weak path regularity conditions. Let us consider the stopping times
$$\tau^j:=\inf\big\{t >0; |W^{(j)}_t|=1 \big\};~j=1,\ldots, p,$$
and we set

$$\psi^{H,j}(u):= \mathbb{E}|\phi^{H,j}_{\tau^ju} - \phi^{H,j}_{0}|^2,~\text{for}~u\ge 0,j=1\ldots,p.$$
Here, if $u$ satisfies $\tau^ju \ge T$ we set $\phi^{H,j}_{\tau^ju}:= \phi^{H,j}_{T}$ and for simplicity we assume that $\psi^{H,j}(0-)=0$.

\begin{theorem} \label{result1}
If $H$ is a $\mathbb{Q}$-square integrable contingent claim satisfying \textbf{(M)} and there exists a representation $\phi^H=(\phi^{H,1},\ldots, \phi^{H,p})$ of $H$ such that $\phi^{H,j}\in \text{B}^2(\mathbb{F})$ for some $j\in \{1,\ldots, p\}$ and the initial time $t=0$ is a Lebesgue point of $u\mapsto \psi^{H,j}(u)$, then

\begin{equation}\label{pwisestrong}
\mathbf{D}^{k,j} X_{T^{k,j}_{1}} \rightarrow  \phi^{H,j}_{0}~\quad \text{as}~k\rightarrow \infty.
\end{equation}
\end{theorem}

\begin{proof}
In the sequel, $C$ will be a constant which may differ from line to line and let us fix $j=1,\ldots,p$. For a given $k\ge 1$, it follows from Lemma \ref{identitytkn} that
\begin{align}\label{withoutshift}
\mathbb{D}^{k,j} X_{T^{k,j}_{1}} & =   \frac{ \mathbb{E} \left[ \int_{0}^{T^{k,j}_{1}} \phi^{H,j}_s dW^{(j)}_s \mid \mathcal{F}^k_{T^{k,j}_{1}} \right]}{ \Delta A^{k,j}_{T^{k,j}_{1}}} \nonumber \\
& = \frac{ \mathbb{E} \left[ \int_{0}^{T^{k,j}_{1}} \left( \phi^{H,j}_s - \phi^{H,j}_{0} + \phi^{H,j}_{0} \right)d W^{(j)}_s \mid \mathcal{F}^k_{T^{k,j}_{1}} \right]}{ \Delta A^{k,j}_{T^{k,j}_{1}}} \nonumber \\
& = \frac{ \mathbb{E} \left[ \int_{0}^{T^{k,j}_{1}} \left(\phi^{H,j}_s - \phi^{H,j}_{0}  \right)dW^{(j)}_s \mid \mathcal{F}^k_{T^{k,j}_{1}} \right]}{ \Delta A^{k,j}_{T^{k,j}_{1}}} + \mathbb{E}  \left[ \phi^{H,j}_{0} \mid \mathcal{F}^k_{T^{k,j}_{1}} \right].
\end{align}
We recall that $T^{k,j}_1 \stackrel{law}{=} 2^{-2k}\tau^j$ so that we shall apply the Burkholder-Davis-Gundy and Cauchy-Schwartz inequalities together with a simple time change argument on the Brownian motion to get the following estimate

\begin{align}\label{espetk1}
\mathbb{E} \Bigg| \frac{\mathbb{E} \left[ \int_{0}^{T^{k,j}_{1}} \left( \phi^{H,j}_s - \phi^{H,j}_{0}\right) d W^{(j)}_s \mid \mathcal{F}^{k}_{T^{k,j}_{1}}  \right]}{\Delta A^{k,j}_{T^{k,j}_{1}}} \Bigg| & \leq 2^k \mathbb{E} \Bigg| \int_{0}^{T^{k,j}_{1}} \left( \phi^{H,j}_s - \phi^{H,j}_{0}\right) d W^{(j)}_s  \Bigg| \nonumber \\
& = 2^k \mathbb{E} \Bigg| \int_{0}^{2^{-2k}} \left( \phi^{H,j}_{\tau^js} - \phi^{H,j}_{0}\right) d W^{(j)}_{\tau^js}\Bigg| \nonumber \\
& \le C2^k \mathbb{E} \Bigg| \int_{0}^{2^{-2k}} \left( \phi^{H,j}_{\tau^js} - \phi^{H,j}_{0}\right)^2\tau^jds  \Bigg|^{1/2} \nonumber \\
& \leq C\mathbb{E}^{1/2} \tau^j \mathbb{E}^{1/2} \frac{1}{2^{-2k}} \int_{0}^{2^{-2k}}\Big( \phi^{H,j}_{\tau^ju} - \phi^{H,j}_{0}\Big)^2 d u \nonumber \\
& = C\mathbb{E}^{1/2} \frac{1}{2^{-2k}} \int_{0}^{2^{-2k}}\Big( \phi^{H,j}_{u\tau^j} - \phi^{H,j}_{0}\Big)^2 du.
\end{align}
\noindent Therefore, the right-hand side of \eqref{espetk1} vanishes if, and only if, $t=0$ is a Lebesgue point of $u\mapsto \psi^{H,j}(u)$, i.e.,

\begin{equation}\label{lpoint}
\frac{1}{2^{-2k}} \int_{0}^{2^{-2k}} \mathbb{E} | \phi^{H,j}_{u\tau^j} - \phi^{H,j}_0 |^2 d u \rightarrow 0~\text{as}~k\rightarrow \infty.
\end{equation}
The estimate \eqref{espetk1}, the limit \eqref{lpoint} and the weak convergence of $\mathcal{F}^{k}_{T^{k,j}_1}$ to the initial sigma-algebra $\mathcal{F}_0$ yield

$$
\lim_{k\rightarrow \infty} \mathbb{D}^{k,j} X_{T^{k,j}_{1}} = \lim_{k\rightarrow \infty} \mathbb{E}  \left[  \phi^{H,j}_{0} \mid \mathcal{F}^k_{T^{k,j}_{1}} \right]  = \phi^{H,j}_{0}$$
strongly in $L^1$. Since $\mathbf{D}^{k,j} X_{T^{k,j}_{1}} = \mathbb{E}\big[\mathbb{D}^{k,j} X_{T^{k,j}_{1}}\big];k\ge 1$ then we conclude the proof.
\end{proof}

\begin{remark}
At first glance, the limit \eqref{pwisestrong} stated in Theorem \ref{result1} seems to be rather weak since it is not defined in terms of convergence of processes. However, from the purely computational point of view, we shall construct a pointwise Monte Carlo simulation method of the GKW projectors in terms of $\mathbf{D}^{k,j}X_{T^{k,j}_1}$ given by~(\ref{exder}). This substantially simplifies the algorithm introduced by Le\~ao and Ohashi~\cite{LEAO_OHASHI09} for the unidimensional case under rather weak path regularity.
\end{remark}

\begin{remark}\label{smoothre}
For each $j=1,\ldots, p$, let us define

 $$\psi^{H,j}(t_0,u):= \mathbb{E}|\phi^{H,j}_{t_0 + \tau^ju} - \phi^{H,j}_{t_0}|^2,~\text{for}~t_0 \in [0,T],~u\ge 0.$$
One can show by a standard shifting argument based on the Brownian motion strong Markov property that if there exists a representation $\phi^H$ such that $u\mapsto \psi^{H,j}(t_0,u)$ is cadlag for a given $t_0$ then one can recover pointwise in $L^1$-strong sense the $j$-th GKW projector for that $t_0$. We notice that if $\phi^{H,j}$ belongs to $\text{B}^2(\mathbb{F})$ and it has cadlag paths then $u\mapsto \psi^{H,j}(t_0,u)$  is cadlag for each $t_0$, but the converse does not hold. Hence the assumption in Theorem~\ref{result1} is rather weak in the sense that it does not imply the existence of a cadlag version of $\phi^{H,j}$.
\end{remark}

\bibliographystyle{plain}
\bibliography{mathf13}

\end{document}